\documentclass[12pt]{article}
\usepackage[scale=0.8]{geometry}
\usepackage[utf8]{inputenc}
\usepackage{amsfonts}
\usepackage{amsmath}
\usepackage{bbm}
\usepackage{amssymb}
\usepackage{amsthm}
\usepackage{mathabx}
\usepackage{mathrsfs}
\usepackage{natbib}
\usepackage{xcolor}
\usepackage{hyperref}
\hypersetup{
    colorlinks=true,
    linkcolor=blue,
    filecolor=magenta,      
    urlcolor=cyan,
    citecolor=blue,
}
\usepackage{graphicx}
\usepackage{subcaption}
\usepackage[plain,noend]{algorithm2e}

\def\vec{\mathop{\rm vec}\nolimits}

\def\tr{\mathop{\rm tr}\nolimits}

\def\dist{\mathop{\rm dist}\nolimits}

\def\amp{\mathop{\:\:\,}\nolimits}
\def\Real{\mathop{\mathbb{R}}\nolimits}

\def\argmin{\mathop{\rm argmin}\nolimits}

\def\tr{\mathop{\rm tr}\nolimits}

\newcommand{\bu}{\boldsymbol{u}}
\newcommand{\bv}{\boldsymbol{v}}

\newcommand{\bx}{\boldsymbol{x}}
\newcommand{\by}{\boldsymbol{y}}

\newcommand{\bA}{\boldsymbol{A}}
\newcommand{\bB}{\boldsymbol{B}}
\newcommand{\bC}{\boldsymbol{C}}
\newcommand{\bD}{\boldsymbol{D}}

\newcommand{\bH}{\boldsymbol{H}}
\newcommand{\bI}{\boldsymbol{I}}

\newcommand{\bR}{\boldsymbol{R}}
\newcommand{\bS}{\boldsymbol{S}}

\newcommand{\bU}{\boldsymbol{U}}
\newcommand{\bV}{\boldsymbol{V}}

\newcommand{\bX}{\boldsymbol{X}}
\newcommand{\bY}{\boldsymbol{Y}}

\newcommand{\bsigma}{\boldsymbol{\sigma}}

\newcommand{\bTheta}{\boldsymbol{\Theta}}

\newcommand{\bSigma}{\boldsymbol{\Sigma}}

\newcommand{\bOmega}{\boldsymbol{\Omega}}
\newcommand{\Ell}{\mathcal{L}}
\newcommand{\Hada}{\circ} 

\usepackage{natbib}

\newtheorem{theorem}{Theorem}
\newtheorem{proposition}{Proposition}
\newtheorem{lemma}{Lemma}

\usepackage{authblk}

\title{A Proximal Distance Algorithm for Likelihood-Based Sparse Covariance Estimation}
 \author[1]{Jason Xu\thanks{Correspondence to:
 \href{mailto:jason.q.xu@duke.edu}{jason.q.xu@duke.edu}}}
 \author[2]{Kenneth Lange}
\affil[1]{Department of Statistical Science, Duke University}
\affil[2]{Department of Computational Medicine, Human Genetics, and Statistics, University of California, Los Angeles}

\date{\vspace{-22pt}}

\begin{document}

\maketitle
\begin{abstract}
This paper addresses the task of estimating a covariance matrix under a patternless sparsity assumption. In contrast to existing approaches based on thresholding or shrinkage penalties, we propose a likelihood-based method that regularizes the distance from the covariance estimate to a symmetric sparsity set. This formulation avoids unwanted shrinkage induced by more common norm penalties and enables optimization of the resulting non-convex objective by solving a sequence of smooth, unconstrained subproblems. These subproblems are generated and solved via the proximal distance version of the majorization-minimization principle. The resulting algorithm executes rapidly, gracefully handles settings where the number of parameters exceeds the number of cases, yields a positive definite solution, and enjoys desirable convergence properties. Empirically, we demonstrate that our approach outperforms competing methods by several metrics across a suite of simulated experiments. Its merits are illustrated on an international migration dataset and a classic case study on flow cytometry. Our findings suggest that the marginal and conditional dependency networks for the cell signalling data are more similar than previously concluded. 
\end{abstract}

\section{Introduction}

The task of estimating a covariance matrix from randomly sampled data is central in multivariate analysis. Unfortunately, estimation is complicated by several statistical and computational obstacles. \textcolor{black}{Chief among the latter is the quadratic growth of the number of free parameters in the number of features $p$. If $n$ is the number of cases, it is known statistically that the sample covariance estimator degrades as the ratio $p/n$ increases \citep{stein1956} and becomes singular as soon as $p>n$.} A more subtle difficulty lies in producing good estimators that maintain positive definiteness. Most approaches seek to mitigate the curse of dimensionality by imposing parsimony through assumptions on the size and structure of the effective parameters, a strategy that has proven successful in many applications. In this paper, we focus on the setting where the covariance matrix follows a patternless sparsity assumption. Here sparsity has an important interpretation; namely, zero entries in the covariance matrix encode marginal independence between features. 

Since the work of \cite{stein1956}, covariance estimation has remained an active area of research. Many regularized estimators have been proposed to achieve sparsity; \cite{pourahmadi2011},
\cite{chi2014stable}, and \cite{fan2016} provide excellent overviews. Some researchers assume a known ordering of variables. To their detriment, such estimators based on tapering, banding, or the Cholesky decomposition generally are sensitive to permutations of the features \citep{wu2003, bickel2008, levina2008, cai2010, bien2016}. When no natural ordering is available, a simple tactic involves \textit{thresholding} the sample covariance matrix by setting small estimated entries to zero \citep{karoui2008, bickel2008covariance, rothman2009, cai2011}. Though such elementwise operations straightforwardly induce sparsity, it is well-documented that the resulting estimator is not always positive definite in finite samples. Related Frobenius norm based approaches include an additional log-barrier term \citep{rothman2012} or appeal to alternating directions methods \citep{xue2012} to enforce positive definiteness. Similar methods have been developed for sparse correlation estimation \citep{cui2016sparse}. In general, great care must be taken in selecting thresholding constants to ensure positive definiteness. In many cases, the appropriate range is too narrow to induce an effective amount of sparsity \citep{azose2016}.

Penalized likelihood techniques offer an alternative to thresholding and are arguably the preferred method for estimating sparse precision (inverse covariance) matrices  \citep{yuan2007,molstad2018}. Sparsity carries a different interpretation here: zero entries in the precision matrix encode conditional rather than marginal independence. In this case, the negative Gaussian log-likelihood is convex, which not only ensures that minimizers are global optima, but enables fast algorithms such as the graphical lasso \citep{friedman2007} that make estimation easy under convex penalties such as an $\ell_1$-norm term. Lasso penalization also comes with disadvantages such as shrinkage toward the origin, which may lead to biased estimates and the inclusion of spurious predictors. 

Penalized likelihood estimation is decidedly more difficult in seeking a sparse covariance matrix. Because the negative log-likelihood in $\bSigma$ is no longer convex, significant computational difficulties arise. These challenges may explain in part the smaller literature \textcolor{black}{on} this task relative to precision estimation. \cite{lam2009} study the properties of $\ell_1$ penalized covariance estimation, and \cite{bien2011} propose a majorization-minimization algorithm that makes use of generalized gradient descent. Under the latter approach, convergence hinges on imposing a Lipschitz differentiability assumption that is realized by restricting the space to a subset of the positive definite cone. In practice, this restriction introduces an additional inner optimization subproblem, which
\textcolor{black}{is more cumbersome to implement} and may be numerically unstable even in moderate dimensions. Step-size selection can precipitate a delicate tradeoff between stability and practical rate of convergence. \cite{azose2016} build upon this work to propose a method for maximum a posteriori estimation that faces similar challenges. They report that cross-validation on a problem with $n=12$ and dimension $p \approx 200$ already becomes computationally impractical.

In the current paper, we revisit the penalized likelihood framework for sparse covariance estimation under a distance-to-set penalty in place of a norm penalty. In prior work, such distance penalties have proven effective in contexts such as generalized linear regression under both rank and sparsity constraints \citep{xu2017generalized}. 
Our penalization keeps parameter estimates close to the sparse constraint set while restricting estimates to the positive definite cone. \textcolor{black}{Neither additional assumptions on the structure of the covariance matrix nor prior knowledge of the location of zero entries is} imposed \citep{chaudhuri2007estimation}. Our method thus performs model selection while delivering a positive definite estimate of the covariance matrix, avoiding the systematic shrinkage engendered by convex norm penalties.  

Distance penalization also confers significant computational advantages. We develop a proximal distance algorithm that effectively solves our  non-convex optimization problem. Like \cite{bien2011}, we \textcolor{black}{employ} the majorization-minimization principle. Our algorithm enjoys a descent property as it converges to a stationary point of the objective, automatically selects a good step-size, and yields closed-form solutions to its subproblems. Our algorithm tends to converge quickly because the underlying surrogate functions tightly approximate the likelihood.  These advantages are illustrated by simulation studies and applications to real data on cell signaling and international migration.

\section{Background and penalized formulation}
Consider estimation of the covariance matrix $\bSigma$ given $n$ independent, identically distributed random vectors $\bX_1, \ldots, \bX_n \sim N_p({\bf 0}, \bSigma)$. Without loss of generality, we focus on the mean zero case and estimation of $\bSigma$ alone. In this scenario, the log-likelihood of the data is
\begin{eqnarray}
\Ell(\bSigma) &= & 
-\frac{n}{2} \ln \det \bSigma -\frac{n}{2} \text{tr} (\bSigma^{-1} \bS),  
\label{eq:ll}
\end{eqnarray}
where $\bS = \frac{1}{n} \sum_{i=1}^n \bX_i \bX_i^T$ denotes the sample covariance matrix. When the data are weakly dependent or non-Gaussian, estimation may still proceed on the basis of $\Ell(\bSigma)$, provided $\Ell(\bSigma)$ is interpreted as a quasi-likelihood.  \textcolor{black}{It is desirable that an estimator of $\bSigma$ be positive definite; previous work has achieved this imposing the domain constraint $\bSigma \succ {\bf 0}$. Alternatively, we may set $\Ell(\bSigma)=-\infty$ whenever $\Sigma$ fails to be positive definite.}

We seek to maximize \eqref{eq:ll} subject to the assumption that many of the entries in $\bSigma$ are zero. Accordingly, let $k$ denote the number of nonzero entries in the upper triangle, and let  $\lVert \bSigma \rVert_0$ denote the total number of nonzero entries in $\bSigma$. Sparse estimation of $\bSigma$ can be formally cast as the constrained optimization problem of minimizing
\begin{eqnarray}\label{eq:direct}
f(\bSigma)& = &\ln \det \bSigma + \text{tr} (\bSigma^{-1} \bS)
 \end{eqnarray}
subject to $\bSigma \succ {\bf 0}$ and $\bSigma$ belonging to the sparsity set
\begin{equation}\label{eq:constraintset}
\mathcal{C} \amp = \amp \left\{ \bSigma \in \Real^{p \times p} : \bSigma = \bSigma^T, \, \lVert \bSigma \rVert_0 \leq 2k+p  \right\}.
\end{equation}
Note here that that diagonal entries of $\bSigma$ are unconstrained.

Directly minimizing criterion \eqref{eq:direct} is challenging. Indeed, letting $m = \binom{p}{2}$, there are $\binom{m}{k}$ different sparsity patterns for a model with $k$ effective parameters. Thus, even ignoring the positive definiteness constraint, optimizing $f(\bSigma)$ over $\mathcal{C}$ quickly becomes combinatorially intractable. As a practical alternative, one can include a lasso penalty regularizing the $\ell_1$ norm of a function of $\bSigma$. Convex relaxation of the $\ell_0$ constraint appearing in $\mathcal{C}$ in this fashion provides a viable means of promoting sparsity by proxy. For covariance estimation, \cite{bien2011} consider such a penalty applied to $\bA \Hada \bSigma$, where $\bA$ has non-negative entries interpretable as weights, and $\Hada$ denotes the Hadamard or element-wise product. The resulting optimization problem
\begin{eqnarray}\label{eq:bien}
\text{minimize} \left\{ \ln \det \bSigma + \tr  (\bSigma^{-1} \bS) + \lambda \lVert \bA \Hada \bSigma \rVert_1  \right\} \text{ subject to \,} \bSigma \succ {\bf 0}  
\end{eqnarray} 
remains nontrivial. This nonconvex objective equals the difference of two convex functions. Exploiting this structure, \cite{bien2011} propose a majorization-minimization (MM) algorithm described in the next section. 

Including a lasso penalty as a proxy for the sparsity constraint entails shrinking the solution globally toward the origin.  Such shrinkage biases parameter estimates toward zero and tends to produce false positives. Nonetheless, several advantages have made the approach popular. Lasso penalties are convex, and their inclusion for solving convex objectives not only admits unique minimizers, but fast algorithms are available for finding their solutions \citep{friedman2007}. 
Unfortunately, as the covariance likelihood is already non-convex, adding an $\ell_1$ penalty in covariance estimation does not yield a convex objective and \textcolor{black}{does not enforce  
positive definiteness}. The remedy of embedding an inner iterative algorithm, such as alternating directions, within an outer gradient descent algorithm is often slow and unstable. Failures of positive definiteness also beset simple thresholding approaches \citep{rothman2009, rothman2012}, and similar remedies in this context are subject to the same criticisms \citep{xue2012}.

As an alternative to solve problem \eqref{eq:direct}, we propose minimizing the penalized objective 
\begin{equation}\label{eq:distobj}
h_\rho(\bSigma) = \ln \det \bSigma + \text{tr} (\bSigma^{-1} \bS) + \frac{\rho}{2} \text{dist}(\bSigma, \mathcal{C} )^2,
\end{equation}
where \textcolor{black}{$\dist(\bSigma, \mathcal{C}) = \inf \{ \|\bSigma - \bA\|:  \bA \in \mathcal{C} \}$} denotes the Euclidean distance from $\bSigma$ to $\mathcal{C}$. The distance penalty pulls the solution $\widehat{\bSigma}$ toward $\mathcal{C}$ and equals zero precisely when $\bSigma \in \mathcal{C}$. This novel formulation now entails solving an unconstrained optimization problem, but coincides with the original objective (2) in the limit as $\rho$ tends to $\infty$. This is summarized in the following restatement of the classical penalty method \citep{courant1943}.

\begin{proposition}\label{prop:clarke} 
Suppose that both the loss $f(\bx)$ and the nonnegative penalty $p(\bx)$ are continuous on $\mathbb{R}^p$ and that the penalized objectives 
\begin{eqnarray*}
F_n(\bx) & = & f(\bx)+\rho_n p(\bx)
\end{eqnarray*}
are coercive on $\mathbb{R}^p$. For any sequence $\rho_n$ increasing to $\infty$, there is a corresponding sequence of minimizers $\bx_n$ with $f(\bx_n) \le f(\bx_{n+1})$. Further, any cluster point of this sequence resides in the feasible region $\mathcal{S}=\{\bx: p(\bx)=0\}$ and attains the minimum value of $f(\bx)$. Finally, if $f(\bx)$ is coercive and possesses a unique minimizer $\bx^\ast$ in $S$, then the sequence $\bx_n$ converges to $\bx^\ast$.
\end{proposition}
This result justifies improving the objective (\ref{eq:distobj}) while gradually increasing the penalty parameter $\rho$ instead of directly tackling the constrained problem \eqref{eq:direct}. We show in the Appendix that positive definiteness of the sample covariance $\bS$ is sufficient to satisfy the technical requirement of coercivity. Notably, coercivity fails to hold when $p>n$, but it can be reintroduced by adding a small multiple of the identity to $\bS$. We later observe that this safeguard is numerically unnecessary in practice, with a negligible difference in performance. In its favor, the distance penalized formulation circumvents explicit consideration of the constraints, evades global shrinkage, and lends itself to the derivation of a practical algorithm via majorization-minimization.

\section{Majorization-Minimization}

Majorization-minimization (MM) algorithms are becoming increasingly popular in solving large-scale optimization problems in statistics and machine learning \citep{mairal2015,lange2016,xu2019power}. A majorization-minimization algorithm successively minimizes a sequence of surrogate functions $g(\bx \mid \bx_k)$ that dominate an objective function $f(\bx)$ and are tangent to it at the current iterate  $\bx_k$. Decreasing $g(\bx \mid \bx_k)$  automatically engenders a decrease in $f(\bx)$, and a local optimum of $f(\bx)$ is found by successively minimizing the sequence of surrogates.  

Majorization requires two conditions: tangency and domination. Formally, these amount to $g(\bx_k~\mid~\bx_k)~= ~f(\bx_k)$ and $g(\bx \mid \bx_k)  \geq f(\bx)$ for every $\bx$. The resulting update $\bx_{k+1}~=~\argmin_{\bx} g(\bx \mid \bx_k)$ implies the string of inequalities
\begin{equation}\label{eq:descent}
  f(\bx_{k+1}) \leq g(\bx_{k+1} \mid \bx_{k}) \leq g(\bx_{k} \mid \bx_{k}) = f(\bx_{k}),
\end{equation}
validating the descent property. Examination of the proof of descent \eqref{eq:descent} shows that exact minimization of $g(\bx \mid \bx_k)$ is not strictly necessary, a practical advantage that we will utilize. The celebrated expectation-maximization (EM) principle \citep{dempster1977} for maximum likelihood estimation is a special case of this principle that relies on the notion of missing data. \textcolor{black}{In this setting the surrogate $g(\bx \mid \bx_k)$ is defined as the expected value of the complete data log-likelihood given the observed data.}

The majorization-minimization principle thus offers a general recipe for converting a hard optimization problem into a sequence of manageable subproblems. The distance-to-set penalty $\text{dist}(\bx, \mathcal{C} )^2$ enters this framework 
through \textit{distance majorization} \citep{chi2014}. The key idea is to write the penalty in terms of the Euclidean norm as
\begin{equation} \label{eq:distance_major}
\text{dist}(\bx, \mathcal{C}) =
\min_{\by \in \mathcal{C}}
\lVert \bx - \by \rVert_2
=\lVert \bx - P_{\mathcal{C}}(\bx) \rVert_2,
\end{equation}
where $P_{\mathcal{C}}(\bx)$ denotes the projection of $\bx$ onto the constraint set $\mathcal{C}$. Squaring the distance term is a practical maneuver that leads to  differentiability and the simple gradient
\begin{equation*}
\nabla \frac{1}{2}\dist(\bx, \mathcal{C})^2 = \bx - P_{\mathcal{C}}(\bx)
\end{equation*}
when $P_{\mathcal{C}}(\bx)$ is single-valued \citep{lange2016}. Fortunately,  the projection operator $P_S(\bx)$ onto any closed set $S$ is single valued except for a set of Lebesgue measure $0$ \citep{keys2019}. Hence, the technical possibility that $P_{\mathcal{C}}(\bx)$ becomes multi-valued for $\mathcal{C}$ non-convex is vanishingly rare from a theoretical perspective. Indeed, this event is negligible in practice as well: if a multi-valued point is encountered, the user is shielded from this exception because the code automatically selects a point in $P_C(x_k)$ and delivers  a valid surrogate.
The distance majorization
$$\text{dist}\left(\bx, \mathcal{C}\right)^2 = \lVert \bx - P_{\mathcal{C}}(\bx) \rVert_2^2 
 \leq \lVert \bx -\by_k \rVert_2^2 \quad \text{ for all }\, \by_k \in P_{\mathcal{C}}(\bx_k)$$ follows directly from the definition of the projection operator 
 $P_{\mathcal{C}}(\bx)$. 
 This majorization is useful in practice because it replaces the distance penalty by a spherically symmetric quadratic with the same gradient at $\bx_k$. 

Recall in the present context, the relevant constraint set $\mathcal{C}$ defined in \eqref{eq:constraintset} consists of all symmetric matrices with at most $k$ nonzero entries in their upper triangle. The choice of $k$ determines the level of sparsity. Computing the projection of a symmetric matrix onto $\mathcal{C}$ is accomplished by setting all but the diagonal and the $k$ largest entries (in absolute value) of each triangle to zero. Because we have not yet imposed the constraint on positive definiteness, this projection is  computed simply by \textcolor{black}{hard thresholding entries in the upper triangle, and propagating the results to the lower triangle symmetrically.}  The next section describes how distance majorization creates a sequence of unconstrained smooth problems, as well as how the positive definiteness constraint can be enforced via simple backtracking.

\section{Algorithm for Sparse Covariance Estimation}
\subsection{A Proximal Distance Algorithm}
The recently introduced proximal distance principle \citep{keys2019} \textcolor{black}{ replaces the constrained problem $\min_{\bx \in \mathcal{C}} f(\bx)$ by unconstrained minimization of the penalized loss $f(\bx)+\frac{\rho}{2}\text{dist}(\bx, \mathcal{C})^2$. In our setting,
the loss (\ref{eq:direct}) plays the role of $f(\bx)$ under the sparsity constraint (\ref{eq:constraintset}).  The unconstrained reformulation  (\ref{eq:distobj}) can then be solved using distance majorization.} Proposition \ref{prop:clarke} implies that if $\rho$ is sufficiently large, then the solution of the penalized problem accurately approximates the solution of the constrained problem. For any given value of $\rho$, applying the proximal distance principle requires majorizing the objective $h_\rho(\bSigma)$ by the function 
$g_\rho(\bx \mid \bx_k )= f(\bx) + \frac{\rho}{2} \lVert \bx - P_\mathcal{C}(\bx_k) \rVert^2$. Now, the minimizer of this surrogate function is given by a proximal operator. Recall that for any function $r(\bx)$, the proximal operator is defined
$$\text{prox}_{\lambda r}(\by) \equiv \underset{\bx}\argmin \,\Big[ r(\bx) + \frac{1}{2\lambda} \lVert \bx - \by \rVert_2^2\Big] $$ with 
$\lambda=\rho^{-1}$. The operator $\text{prox}_{\lambda r}(\by)$ represents a compromise between minimizing $r(\bx)$ and hewing toward $\by$, with the parameter $\lambda$ modulating the tradeoff; \cite{polson2015} provide an excellent overview of proximal methods in statistics.

\textcolor{black}{Because it is not possible to find an analytic expression for the proximal operator of $f(\bx)$, we cannot easily minimize the surrogate function $g_\rho(\bSigma \mid \bSigma_k)$ generated by the covariance likelihood. Instead, we construct a more useful surrogate function $q_\rho(\bSigma \mid \bSigma)$ using a local quadratic approximation tailored to $g_\rho(\bSigma \mid \bSigma_k)$. } This local surrogate possesses three advantages over linear surrogates used in past approaches \citep{bien2011,azose2016}. First, a quadratic surrogate provides a tighter approximation than the linear surrogates previously applied in this problem, often translating to dramatically more efficient steps toward the optimum. \textcolor{black}{Second, the proximal operator of our surrogate} $q_\rho(\bSigma \mid \bSigma_k)$ admits a closed form solution; each subproblem can be minimized exactly, in contrast to gradient steps whose progress and stability depends heavily on the choice of step sizes. Finally, by exploiting a surprising connection to control theory, evaluation of these closed solutions becomes practical, effecting a reduction in computational complexity from $\mathcal{O}(p^6)$ using a na\"ive evaluation to a more tractable $\mathcal{O}(p^3)$. 

\subsection{\textcolor{black}{Constructing and Minimizing the Surrogate}}

Recall that the relevant loss is $f(\bSigma)=\ln \det \bSigma + \text{tr} (\bSigma^{-1} \bS)$. To define a sequence of quadratic approximations to 
$f(\bSigma)$, we take matrix directional derivatives of $f(\bSigma)$ in the directions $\bU$ and $\bV$,
\begin{align*}
d_{\bV} f(\bSigma) &= \tr (\bSigma^{-1} \bV) - \tr (\bSigma^{-1} \bV \bSigma^{-1} \bS) \\
d_{\bU} d_{\bV} f(\bSigma) &= -\tr (\bSigma^{-1} \bU \bSigma^{-1} \bV) + \tr (\bSigma^{-1} \bU \bSigma^{-1} \bV \bSigma^{-1} \bS ) + \tr (\bSigma^{-1} \bV \bSigma^{-1} \bU \bSigma^{-1} \bS ) \\
 \bV^T d^2 f(\bSigma) \bV &= -\tr( \bSigma^{-1} \bV \bSigma^{-1} \bV) + 2 \tr (\bSigma^{-1} \bV \bSigma^{-1} \bV \bSigma^{-1} \bS ) ,
\end{align*}
where the quadratic form in the last line is obtained by setting $\bU=\bV$. The second differential simplifies considerably if we replace $\bS$ by its expected value $\mathbb{E}(\bS) = \bSigma$. This maneuver is familiar from the derivation of Fisher's scoring algorithm. This substitution precipitates a cancellation of higher order terms; the overall result 
 \[ \bV^T d^2 f(\bSigma)\bV  \approx  \tr( \bSigma^{-1} \bV \bSigma^{-1} \bV) \]
is a positive definite quadratic form. We may now define an approximate quadratic surrogate 
$q_\rho(\bSigma \mid \bSigma_k)$ of the penalized objective $h_\rho(\bSigma)$ by taking a second-order Taylor expansion of the loss $f(\bSigma)$ about the current estimate, namely
\begin{align*}
 q_\rho( \bSigma \mid \bSigma_k ) &= f(\bSigma_k) + \tr[ \bSigma_k^{-1} (\bSigma - \bSigma_k) ] - \tr[ \bSigma_k^{-1} \bS \bSigma_k^{-1} (\bSigma - \bSigma_k ) ] \\
 & \qquad + \frac{1}{2} \tr [ \bSigma_k^{-1} (\bSigma - \bSigma_k) \bSigma_k^{-1} (\bSigma - \bSigma_k) ] + 
 \frac{\rho}{2} \lVert \bSigma - P_\mathcal{C}(\bSigma_k) \rVert_F^2.
 \end{align*}
Here the majorized distance penalty appears as the final term. In contrast to an $\ell_1$-penalized loss, this surrogate is differentiable with gradient expression
\[ \frac{d}{d\bSigma}q_\rho(\bSigma\mid\bSigma_k) =   \bSigma_k^{-1} - \bSigma_k^{-1} \bS \bSigma_k^{-1} + \bSigma_k^{-1} (\bSigma - \bSigma_k) \bSigma_k^{-1} +\rho [ \bSigma - P_\mathcal{C}(\bSigma_k)] . \]
Equating the gradient to $\bf 0$ and rearranging yields the stationarity equation 
\begin{equation}\label{eq:sylvester}
 \rho P_\mathcal{C}(\bSigma_k) + \bSigma_k^{-1} \bS \bSigma_k^{-1} = \rho \bSigma + \bSigma_k^{-1} \bSigma \bSigma_k^{-1}. \end{equation}
If we abbreviate the left-hand side by
\[ \bC_k \equiv \rho P_\mathcal{C}(\bSigma_k) + \bSigma_k^{-1} \bS \bSigma_k^{-1}  \]  
and stack matrices into vectorized notation, then equation \eqref{eq:sylvester} can be rewritten
\[ \vec(\bC_k) = \rho \vec(\bSigma) + (\bSigma_k^{-1} \otimes \bSigma_k^{-1}) \vec(\bSigma), \]
where  $\otimes$ denotes the Kronecker product. Upon inversion, the solution amounts to
\begin{equation}\label{eq:sylvsol}
 \vec(\widehat{\bSigma}) = \left[\rho \bI_{p^2}  + (\bSigma_k^{-1} \otimes \bSigma_k^{-1}) \right]^{-1} \vec(\bC_k),
\end{equation}
and we may recover the minimizer $\widehat{\bSigma}$ by reshaping. 
The analytic solution \eqref{eq:sylvsol} involves the inverse of a
$p^2 \times p^2$ matrix and hence scales as $\mathcal{O}(p^6)$.
This computational load puts problems with even moderate dimension $p$ beyond reach.
Upon multiplying both sides by the constant $\bSigma_k$ and closer inspection, equation \eqref{eq:sylvester} takes the general form 
$\bA\bSigma + \bSigma\bB = \bC$, which we recognize as a Sylvester equation in $\bSigma$. Like the closely related and better-known Lyapunov equations arising in dynamical systems, Sylvester equations are well-studied in control theory and eigenvalue problems \citep{higham2002}. It is known that the equation has a unique solution if and only if the eigenvalues of $\bA$ and $-\bB$ are distinct; this condition holds in the present case because $\bSigma_k^{-1}$ is positive definite. More pertinently, we can borrow a numerical method from the control theory literature. An algorithm due to \cite{bartels1972} provides a more efficient solution than direct evaluation of equation \eqref{eq:sylvsol}. 
The first step and crux of the procedure lies in transforming the problem into Schur form by computing decompositions $\bA = \bU \bR \bU^T$ and $\bB = \bV \bS \bV^T$ via the QR algorithm. Because $\bR$ and $\bS$ are upper triangular, the equivalent upper triangular system $\bR \bY + \bY \bS = \bU^T \bC \bV$ with $\bY = \bU^T \bSigma \bV$ can be solved by simple back-substitution. Multiplication then recovers the original solution $\bSigma = \bU \bY \bV^T$. The computational complexity declines from $\mathcal{O}(p^6)$ operations required to compute formula \eqref{eq:sylvsol} to $\mathcal{O}(p^3)$. 
Current state-of-the-art implementations are variations on this theme and possess the same overall complexity; see \cite{simoncini2016} for details.

Before proceeding, we briefly mention that sparse fitting of the sample correlation matrix can exploit the same algorithm with a simple modification that has been noted previously in the literature. Let $\bR = \bD^{-1/2} \bS \bD^{-1/2}$ denote the sample correlation matrix, where $\bD = \text{diag}(\bS)$ contains the observed variances. In estimation with $\bR$ replacing 
$\bS$, we minimize the criterion  
$$\ln\det \bTheta + \tr(\bTheta^{-1}\bR ) + \frac{\rho}{2} 
\text{dist}(\bTheta,\mathcal{CR})^2$$ 
over $\bTheta \succ {\bf 0}$, where $\mathcal{CR}$ is the set of $k$-sparse
symmetric matrices with unit diagonal entries. Projection of $\bTheta_k$ onto this set maps the diagonal entries of $\bTheta_k$ to $1$ and treats the off-diagonal entries as before.

\subsection{Positive Definiteness and Gradient Interpretation}
So far, the penalty term in the objective \eqref{eq:distobj} only accounts for the sparsity set constraint $\mathcal{C}$. Because the matrix $\rho P_\mathcal{C}(\bSigma_k) + \bSigma_k^{-1} \bS \bSigma_k^{-1}$ is not guaranteed to be positive definite, neither is the solution $\widehat\bSigma$ that minimizes the surrogate given in equation \eqref{eq:sylvester}. Moreover, the approximate surrogate $q_\rho(\bSigma \mid \bSigma_k)$ does not strictly majorize $h_\rho(\bSigma) =f(\bSigma)+\rho/2 \cdot \dist(\bSigma, \mathcal{C})^2$ for all possible $\bSigma$, and so na\"ively minimizing $q_\rho(\bSigma \mid \bSigma_k)$ does not necessarily decrease $h_\rho(\bSigma)$. Both of these issues can be handled gracefully via backtracking. The next proposition ensures the success of step-halving, which amounts to defining 
\begin{equation}\label{eq:backtrack}
\bSigma_{k+1} = \bSigma_k + \frac{1}{2^s}(\widehat\bSigma - \bSigma_k) 
\end{equation} 
based on the smallest integer $s \ge 0$ that renders $\bSigma_{k+1} \succ {\bf 0}$ and decreases the objective $h_\rho(\bSigma)$. 
The result becomes clear after considering the representation
\begin{equation}\label{eq:newton}
\widehat\bSigma = \bSigma_k +\bv_k 
= \bSigma_k -\bH_k^{-1} \nabla q_\rho(\bSigma_k \mid \bSigma_k),
\end{equation}
where $\bH_k$ is the scoring approximation to the Hessian; a complete proof of the following proposition appears in the Appendix.
\begin{proposition} \label{step_halving_prop}
If $\bSigma_k$ is not a stationary point of $h_\rho(\bSigma)$, then there exists an integer $s$ such that $\bSigma_{k+1}$ given in equation \eqref{eq:backtrack} satisfies  $\bSigma_{k+1} \succ {\bf 0}$ and $h_\rho(\bSigma_{k+1}) < h_\rho(\bSigma_k)$ .
\end{proposition}

\noindent \textcolor{black}{The overall method is summarized in pseudocode in Algorithm \ref{alg:1}, which reveals that the careful technical work behind the preceding analysis is largely hidden from the user's perspective. The resulting algorithm is relatively transparent and easy to implement.}

\begin{algorithm}[h]
\DontPrintSemicolon
  \KwData{Sample covariance $\bS$, sparsity level $K$, parameter $\rho>0$, tolerance $\epsilon>0$. }
  
  Initialize $i=1$, $\bSigma_i = \text{Diag}(\bS)$, $h_0 = 0,$ and $h_i = 1$. 
  
  \While{$ | \frac{h_{i} - h_{i-1}}{h_{i-1}} | > \epsilon$}
    {
        Compute $\hat\bSigma$ via solving Sylvester equation or  \eqref{eq:sylvsol}, with $\bC_i = \rho P_\mathcal{C}(\bSigma_i) + \bSigma_i^{-1} \bS \bSigma_i^{-1},$ where $P_\mathcal{C}(\bSigma) $ sets all but the $K$ largest magnitude entries of $\bSigma$ to zero.
        
        Set $\bSigma_{i+1} = \hat\bSigma$, $j=1$ \tcp*{check that $\bSigma_{i+1} \succeq \mathbf{0}$, else backtrack }
        
        \While{$\mathbf{0} \succ \bSigma_{i+1}$}  
        {        
        $\bSigma_{i+1} = \bSigma_i + \frac{1}{2^j}(\widehat\bSigma - \bSigma_i)$ 
        
        $j = j+1$
        }
        
        Update objective $h_{i+1} = \ln \det \bSigma_{i+1} + \text{tr} (\bSigma_{i+1}^{-1} \bS) + \frac{\rho}{2} \| P_\mathcal{C}(\bSigma_{i+1}) - \bSigma_{i+1} \|^2 $
        
        Update $i = i+1$; $\rho = \rho \cdot 1.2$

    }
    
    \textbf{Return:} $\bSigma_i$
\vspace{10pt}
\caption{\textcolor{black}{Example pseudocode  implementation of proximal distance algorithm. }}\label{alg:1}
\end{algorithm}

Before proceeding further, let us pause to compare our surrogate to the surrogate proposed in the sparse covariance method of \cite{bien2011}. Based on the concave-convex procedure of \cite{yuille2003}, they employ the tangent plane majorizer
\[ \ln \det \bSigma_k + \text{tr} (\bSigma_k^{-1} \bSigma) - p + \text{tr} (\bSigma^{-1} \bS) + \lambda \lVert \bA \Hada \bSigma \rVert_1 \]
to the $\ell_1$-penalized objective \eqref{eq:bien}. 
The resulting majorization-minimization iteration
\[ \bSigma_{k+1} = \underset{\bSigma \succ {\bf 0}} \argmin \left[ \text{tr} (\bSigma_{k}^{-1} \bSigma) + \text{tr} (\bSigma^{-1} \bS) + \lambda \lVert \bA \Hada \bSigma \rVert_1 \right] \]
is carried out via generalized gradient descent \citep{beck2009}. The choice of a good step-size is crucial for a reasonable rate of convergence in practice. Because the linear approximation only loosely models their objective function
\[\ln \det \bSigma + \text{tr} (\bSigma^{-1} \bS) + \lambda \lVert \bA \Hada \bSigma \rVert_1, \] 
a given step-size may be well-suited at some points but may drastically overshoot the minimum or exit the positive definite cone at others. Whenever the latter occurs, an additional subproblem must be solved by an alternating directions method \citep{boyd2011}. This inner optimization loop slows
convergence and is decidedly more difficult to implement than backtracking. Stability can be enhanced by decreasing the initial step-size at the expense of more outer iterations. Our quadratic expansion of the log-likelihood produces an approximate surrogate $q_\rho(\bSigma \mid \bSigma_k)$ that hugs our objective $h_\rho(\bSigma)$ more closely. Substitution of a distance penalty for a lasso penalty also enjoys smoothness. \textcolor{black}{In practice, the minimizer
$\widehat\bSigma$ of $q_\rho(\bSigma \mid \bSigma_k)$ rarely fails to diminish $h_\rho(\bSigma)$ or stay within the positive definite cone, so typically we update $\bSigma_{k+1} = \widehat\bSigma$ without backtracking}.  These differences translate to substantial performance advantages, as illustrated in Section \ref{sec:results}.

\textcolor{black}{Finally, an anonymous reviewer has raised the question of  convergence to the global optimum, a valid concern that besets all non-convex optimization problems. To our disappointment, we initially found that our algorithm was somewhat sensitive to initial guesses $\bSigma$ close to the sample covariance matrix $\bS$.  After some experimentation we discovered that the algorithm delivers remarkably stable performance when initiated instead as a diagonal matrix, with sample variances appearing along the diagonal. It may be that perturbations of the full $\bS$ are more likely to lie close to the constraint boundary. This can impede progress and precipitate smaller gradient steps when more backtracking is necessary. Starting from a diagonal or even identity matrix, the algorithm tends to stay well within the interior of the positive definite cone, and has more room to learn from the data and in turn consistently reach the optimum.}

\subsection{Convergence}
Recall that due to non-convexity of the symmetric sparsity set $\mathcal{C}$, it is possible that there exist exceptional points at which the projection operator $P_\mathcal{C}(\bSigma)$ is multi-valued.  The penalty $\text{dist}(\bSigma,\mathcal{C})^2$ and in turn the objective $h_\rho(\bSigma)$ are differentiable where $P_\mathcal{C}(\bSigma_k)$ is single-valued, but merely semi-differentiable elsewhere. In contrast, the surrogate $q_\rho(\bSigma \mid \bSigma_k)$ is differentiable regardless of the projected point selected from $P_\mathcal{C}(\bSigma_k)$. 

Although standard convergence results for gradient methods and majorization-minimization algorithms do not immediately apply in proving convergence \citep{lange2016}, theoretical guarantees can be established by appealing to the general theory of \cite{zangwill}, which encompasses continuous objectives and multi-valued algorithm maps. \textcolor{black}{ One can represent our method as an algorithm map $A(\bSigma)$ taking the current iterate $\bSigma_k$ to the next iterate $\bSigma_{k+1}$. Our novel analysis below treats $A$ as a set-valued map to fully account for the technical possibility that the projection operator is multi-valued, even though the set of points where this can occur has measure zero. The following global convergence result is proved in the Appendix.} 
\begin{theorem}\label{thm:convergence}
Consider the sequence $\bSigma_{k+1} = \bSigma_k + \eta_k \bv_k \in A(\bSigma_k)$ generated by the search direction $\bv_k$ of equation \eqref{eq:newton} and the step length $\eta_k = \argmin_{\eta \in [0,1]} h_\rho(\bSigma_k+\eta \bv_k)$. If the initial point $\bSigma_0$ is positive definite and the sample covariance matrix $\bS$ is nonsingular, then the sequence $\bSigma_k$ is bounded and falls within the interior of the positive definite cone. Furthermore, all of its limit points are stationary points of $h_\rho(\bSigma)$. 
\end{theorem}
Though the result suggests promising performance despite non-convexity, we discuss several limitations. First, to simplify mathematical analysis, it supposes an exact line search. This assumption can be relaxed at the expense of a more complicated proof. 
Second, although the algorithm invariably converges, the proposition cannot guarantee convergence to a global minimizer. \textcolor{black}{It simply says that a convergent subsequence exists whose limit is a stationary point $\bSigma$ satisfying 
\begin{equation} \label{stat_set}
\nabla g_\rho(\bSigma \mid \bSigma) = \bSigma^{-1} - \bSigma^{-1}\bS\bSigma^{-1}
+\rho(\bSigma-\bTheta) = {\bf 0}
\end{equation}
for some $\bTheta \in P_\mathcal{C}(\bSigma)$. As we expect, this stationarity condition is necessary for $\bSigma$ to furnish a global minimum. Indeed, if it fails, we take
$\bTheta \in P_\mathcal{C}(\bSigma)$ with $\nabla g_\rho(\bSigma \mid \bSigma) \ne {\bf 0}$. Then the negative gradient $-\nabla g_\rho(\bSigma \mid \bSigma)$ is a descent direction for $g_\rho(\bSigma \mid \bSigma)$, which majorizes $h_\rho(\bSigma)$. Hence, $-\nabla g_\rho(\bSigma \mid \bSigma)$ would also be a descent direction for $h_\rho(\bSigma)$, contradicting even  local optimality of $\bSigma$. Leveraging majorizing surrogates in this fashion establishes directional stationarity, the strongest kind of stationarity in semidifferentiable optimization \citep{pang2017computing,cui2018composite}, while avoiding  the complications that often come with checking the condition explicitly.}

\section{Empirical results}\label{sec:results}
\subsection{Simulation study}\label{sec:sim}

We illustrate the practical merits of our method on a suite of simulated examples. An open-source Julia implementation of the algorithm is available at the first author's website. 
In all examples, we initialize our algorithm from the diagonal matrix of sample variances. In practice, taking $\bSigma_0 = \bS$ leads to excessive backtracking in some runs. We initialize $\rho$ at $0.1$ and increase it by a factor of $1.2$ each iteration. Convergence is declared based on a relative tolerance of $10^{-6}$.

\begin{figure}[htbp]
\centering 
\hspace{-4pt}
\includegraphics[width = .32\textwidth]{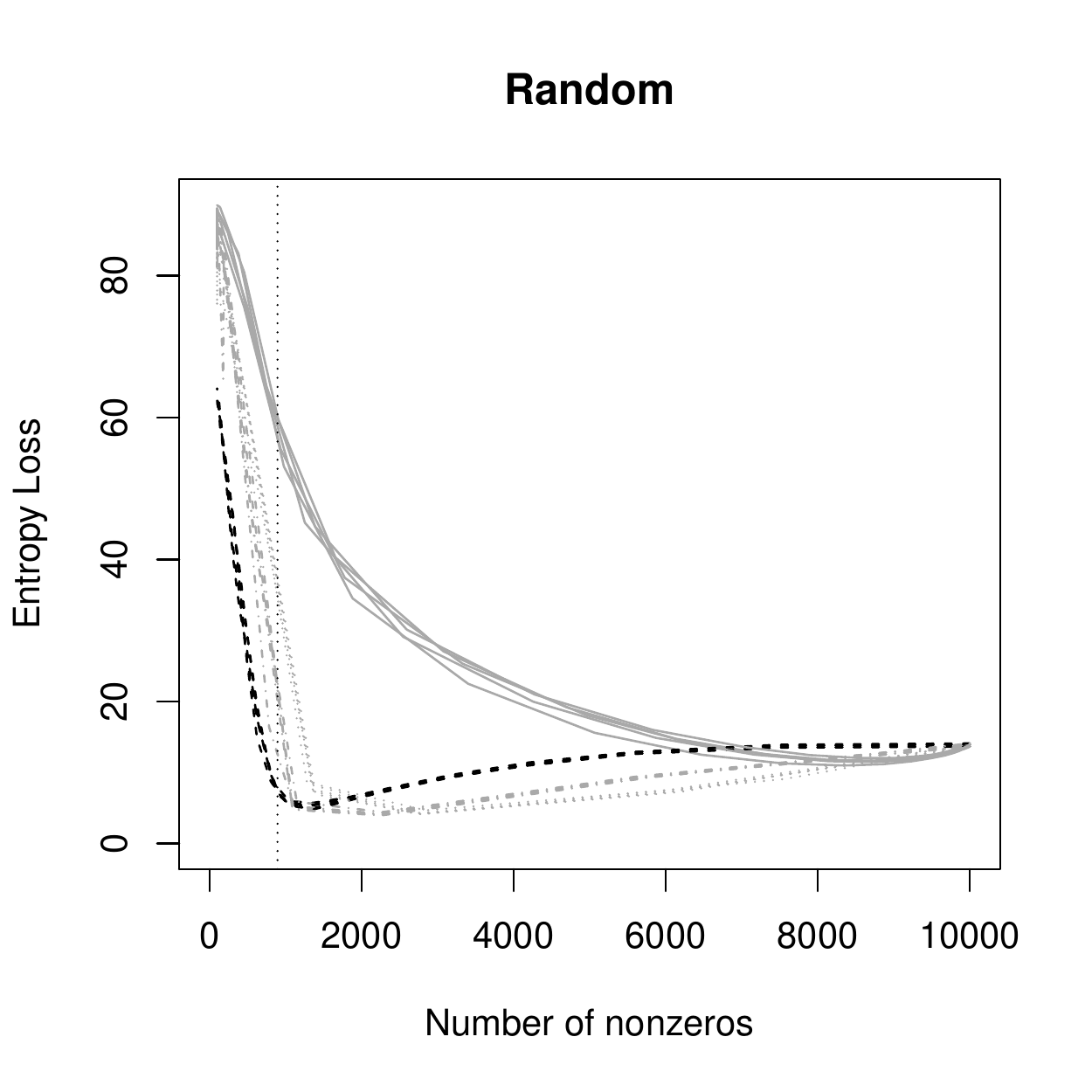} \hspace{-4pt}
\includegraphics[width = .32\textwidth]{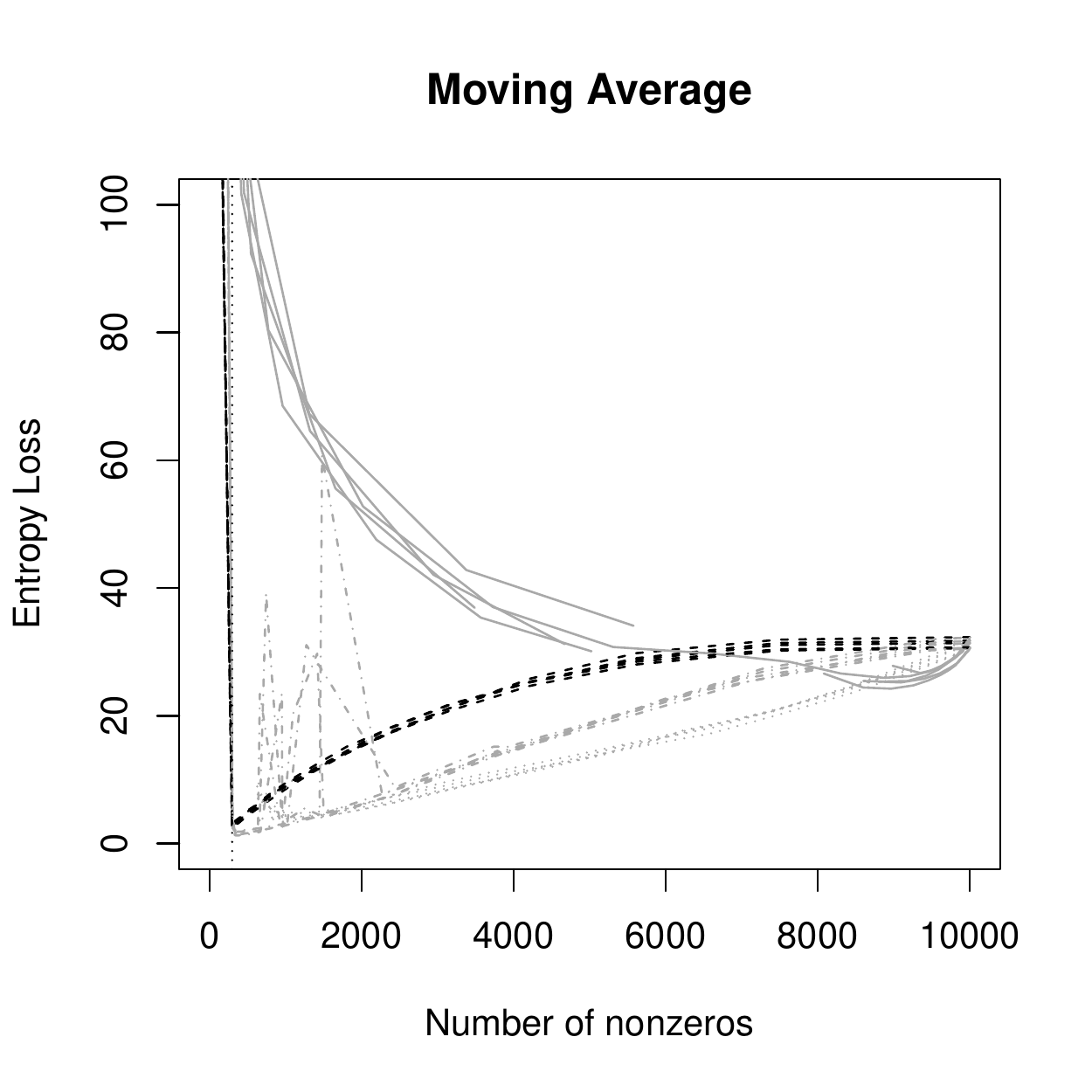} \hspace{-4pt}
\includegraphics[width = .32\textwidth]{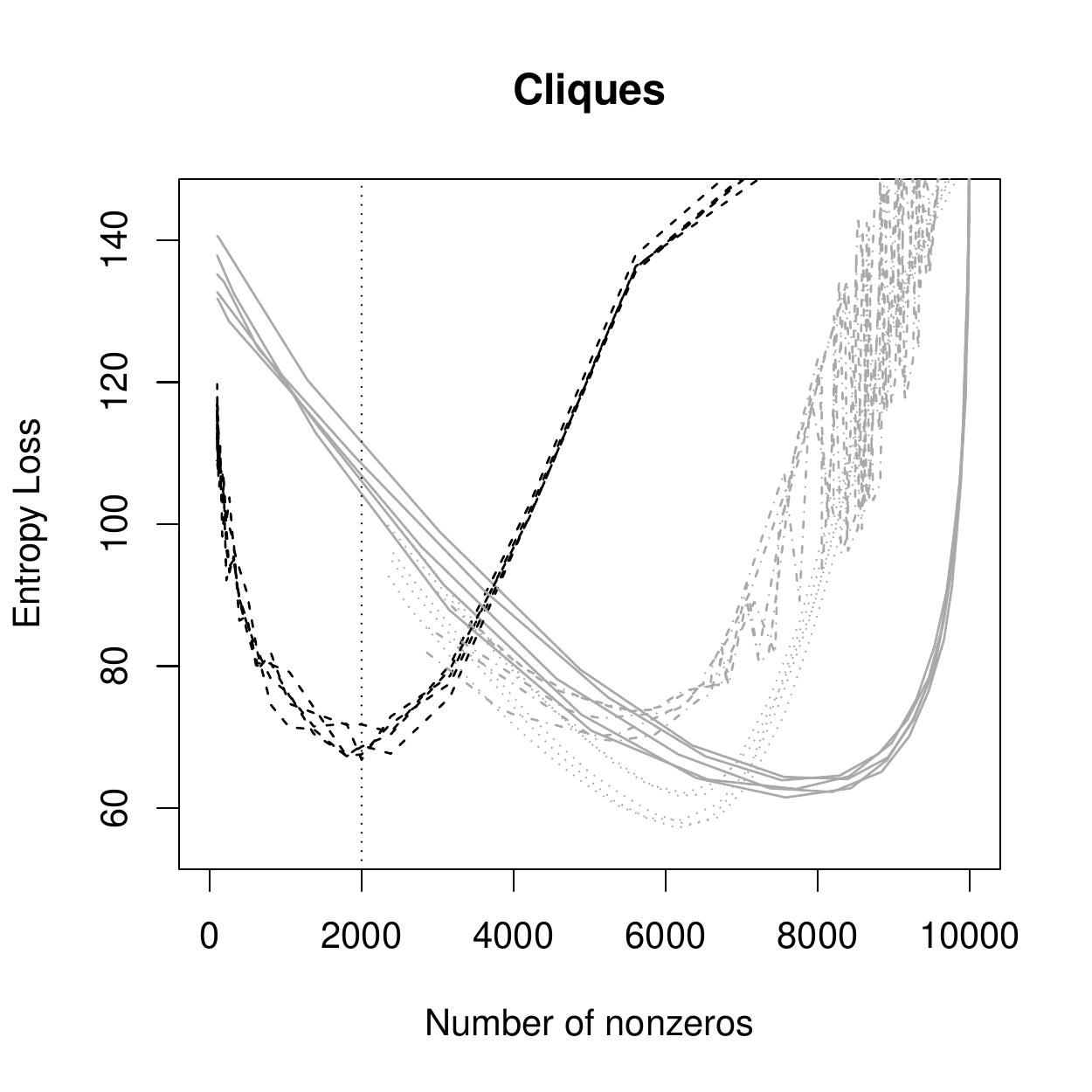} \hspace{-2pt}
\caption{Entropy loss of estimates under each method is plotted over $5$ repeat trials. The black dashed lines represent proximal distance results. Gray solid, dotted, and dot-dash lines represent thresholding, generalized gradient, and ``adaptive" generalized gradient respectively. Vertical line marks true number of nonzeros.} 
\label{fig:sim}
\end{figure}

Figure \ref{fig:sim} \textcolor{black}{summarizes  results for a synthetic data design taken from a study by \cite{bien2011}; further details also appear in the Appendix. Each of the three variants of the underlying model (independent, moving average, and cliques) exhibits $8 \%$ percent sparsity (nonzero entries) with $p > n$.} Following their analysis, we report performance as measured by the entropy loss $\tr (\bSigma^{-1} \hat\bSigma ) - \log \det ( \bSigma^{-1} \hat\bSigma ) - p $, and under receiver operating characteristic (ROC) curves \textcolor{black}{in the Appendix}. Previous authors recommend entropy loss as a measure when the covariance matrix is the primary object of interest \citep{huang2006covariance,levina2008}; note the role of $\bSigma$ in the entropy loss is analogous to how $\bOmega = \bSigma^{-1}$ enters the Kullback-Leibler loss.

Figure \ref{fig:sim} \textcolor{black}{shows a clear performance advantage of the} proximal distance algorithm that becomes more pronounced in the high-dimensional settings. In reproducing the results of \cite{bien2011}, we confirm that calls to the alternating directions method to enforce \textcolor{black}{positive definiteness} are relatively rare when $n>p$. However, this is not the case in the high-dimensional regime where the sample covariance \textcolor{black}{is not of} full rank. In our experience, numerical errors arise in switching between generalized gradient steps and alternating direction method corrections. \textcolor{black}{The results depicted in Figure \ref{fig:sim} reveal some of this instability over five random replicate trials}, most notably for the adaptive version of the generalized gradient which uses reciprocals of the entries in the sample covariance as weights in the penalty. At best, it is necessary to significantly reduce step size, resulting in \textcolor{black}{slower} progress. 

Next, we provide a detailed comparison to the soft and hard thresholding methods of \cite{bickel2008covariance} under the optimal thresholding \textcolor{black}{suggested by \cite{fang2016tuning}, as well as the penalized log-barrier method of \cite{rothman2012}}.  We omit generalized gradient descent \citep{bien2011} in this second  study due to \textcolor{black}{its excessive runtimes under cross-validation, and note that the method of \cite{xue2012} reports nearly identical performance to the penalized log-barrier method we consider}. We evaluate the entropy loss and root mean-squared error between $\bSigma$ and $\hat\bSigma$, and report false positive and false negative rates for identifying the nonzero entries in $\bSigma$. Results are presented in Tables \ref{tab:false} through \ref{tab:rmse}.

\begin{table}[htbp] \centering
\begin{tabular}{ l | c | c | c | c  }
p & Proximal Distance & Soft Threshold & Hard Threshold & Log Barrier \\ \hline
20 & 0.1 / 0.0 & 18.9 / 0.0 & 0.1 / 0.3 & 7.8 / 0.0 \\
30 & 0.2 / 0.6 & 12.4 / 0.0 & 0.2 / 2.7 & 5.8 / 0.0 \\
50 & 0.4 / 1.9 & 8.4 / 0.0 & 0.1 / 7.2 & 3.9 / 0.3 \\ 
100 & 0.5 / 17.8 & 4.4 / 8.5 & 0.1 / 42.0 & 3.2 / 9.7\\
200 & 1.0 / 42.4 & 4.3 / 33.2 & 0.0 / 79.7 &  1.1 / 50.1 \\ 
\end{tabular} 
\caption{Percentages of false positives
(on the left) and false negatives (on the right) over $50$ replications. All methods are tuned using $5$-fold cross validation. Standard errors are omitted; the largest standard error in column $1$ was $0.027$. }
\label{tab:false}
\end{table}

\begin{table}[htbp] \centering
\begin{tabular}{ l | c | c | c  | c }
$p$ & Proximal Distance & Soft Threshold  & Hard Threshold & Log Barrier \\ \hline
20 & 0.28 (0.09) & 0.94 (0.20) & $0.49$ (0.43) & 2.01 (0.6)\\
30 & 0.61 (0.27) & 2.35 (0.63) & $1.40$ (0.98) & 4.6 (0.9) \\ 
50 & 2.11 (0.81) & 6.24 (1.16) & $5.48$ (2.59) & 11.7 (1.2) \\
100 & 17.6 (3.3) & 28.7 (2.6) & $43.7$ (10.2) & 42.6 (4.1) \\
200 & 119.6 (6.3) & 140.3 (4.4) & $206.3$ (6.6) & 179.8 (5.1)
\end{tabular} 
\caption{Comparison of methods in terms of entropy loss over $50$ replications. All methods are tuned using $5$-fold cross validation. Standard errors are in parentheses. }
\label{tab:entropyloss}
\end{table}

\begin{table}[htbp] \centering
\begin{tabular}{ l | c | c | c | c }
$p$ & Proximal Distance  & Soft Threshold & Hard Threshold & Log Barrier \\ \hline
20 & 0.050 (0.011) & 0.092 (0.012) & 0.062 (0.015) & 0.078 (0.011) \\
30 & 0.061 (0.012) & 0.102 (0.009) & 0.073 (0.016) & 0.085 (0.008) \\
50 & 0.081 (0.011) & 0.096 (0.006) & 0.079 (0.010) & 0.080 (0.008) \\
100 & 0.118 (0.008) & 0.124 (0.005) & 0.128 (0.007) & 0.118 (0.004) \\
200 & 0.141 (0.005) & 0.148 (0.002) & 0.150 (0.001) & 0.143 (0.002) \\ 
\end{tabular} 
\caption{Comparison of methods in terms of root mean-square error over $50$ replications. All methods are tuned using $5$-fold cross validation. Standard errors are in parentheses. }
\label{tab:rmse}
\end{table}

 We vary the number of features from $p=20$ to $200$. Under each setting, $50$ replicate datasets of size $n=100$ are generated using a true covariance matrix with $2\%$ sparsity. \textcolor{black}{We remark that when $k$ is known, it can be directly specified in our method. In contrast hyper-parameter tuning remains necessary with known $k$ under shrinkage penalties. Nonetheless, we select $k$ as well as the tuning constants of competing methods by $5$-fold cross-validation to allow a generous comparison, with complete details in the Appendix}. Table \ref{tab:entropyloss} shows that the proximal distance algorithm achieves lower average entropy loss than thresholding and the log-barrier penalized method, a trend that is also conveyed \textcolor{black}{by the root mean-squared error comparisons in Table \ref{tab:rmse}.}
 
 Table \ref{tab:false} shows that hard thresholding typically offers the lowest false positive rates,  often at the expense of an alarmingly high false negative rate. In contrast, our method offers a comparable false positive rate while introducing strikingly  fewer false negatives. As expected, soft thresholding introduces many false positives in all cases. \textcolor{black}{The log-barrier penalized approach shows a qualitatively similar trend to soft thresholding, but tends to strike a better balance, exhibiting a noticeably lower false positive rate at the cost of a minor increase in false negatives. As we increase $p$, both soft thresholding and the log-barrier penalized method begin to suffer a comparable false negative rate to the proximal distance algorithm.} It is notable that even when $p$ is small, the existing methods introduce a nontrivial number of either false positives or false negatives, while the proximal distance algorithm can maintain a low rate on both fronts. \textcolor{black}{Finally, it is worth noting that while our proposed method is a non-convex formulation, we did not observe the algorithm stopping short at local minima. For a fixed synthetic dataset, perturbing the initial guess over twenty trials consistently delivered the same optimum. In the results reported above, we run one instance of each algorithm per simulated dataset. Taking the best of several random restarts would only result in more favorable performance of the proximal distance method in the possibility that it converged to inferior local optima in some trials.}
 
\textcolor{black}{While we omit a detailed runtime comparison due to differences in implementations across programming languages, we report average runtimes of our proposed method as $p$ increases beyond the scope of the previous simulations. Figure \ref{fig:runtime} reveals that for the largest case we consider with $p=5000$, in which there are tens of millions of free parameters under the pattern-less sparsity assumption, the problem remains tractable with a runtime of under two and a half hours on a standard laptop. Most settings complete in seconds, and the right panel shows that the runtime scales roughly as $p^3$. In contrast, \cite{xue2012} report that the log-barrier method becomes unwieldy for $p>200$, while existing likelihood-based methods such as the generalized gradient method in the first simulation study are even slower by a large margin.}
\begin{figure}[htbp]
\centering 
\includegraphics[width = .48\textwidth]{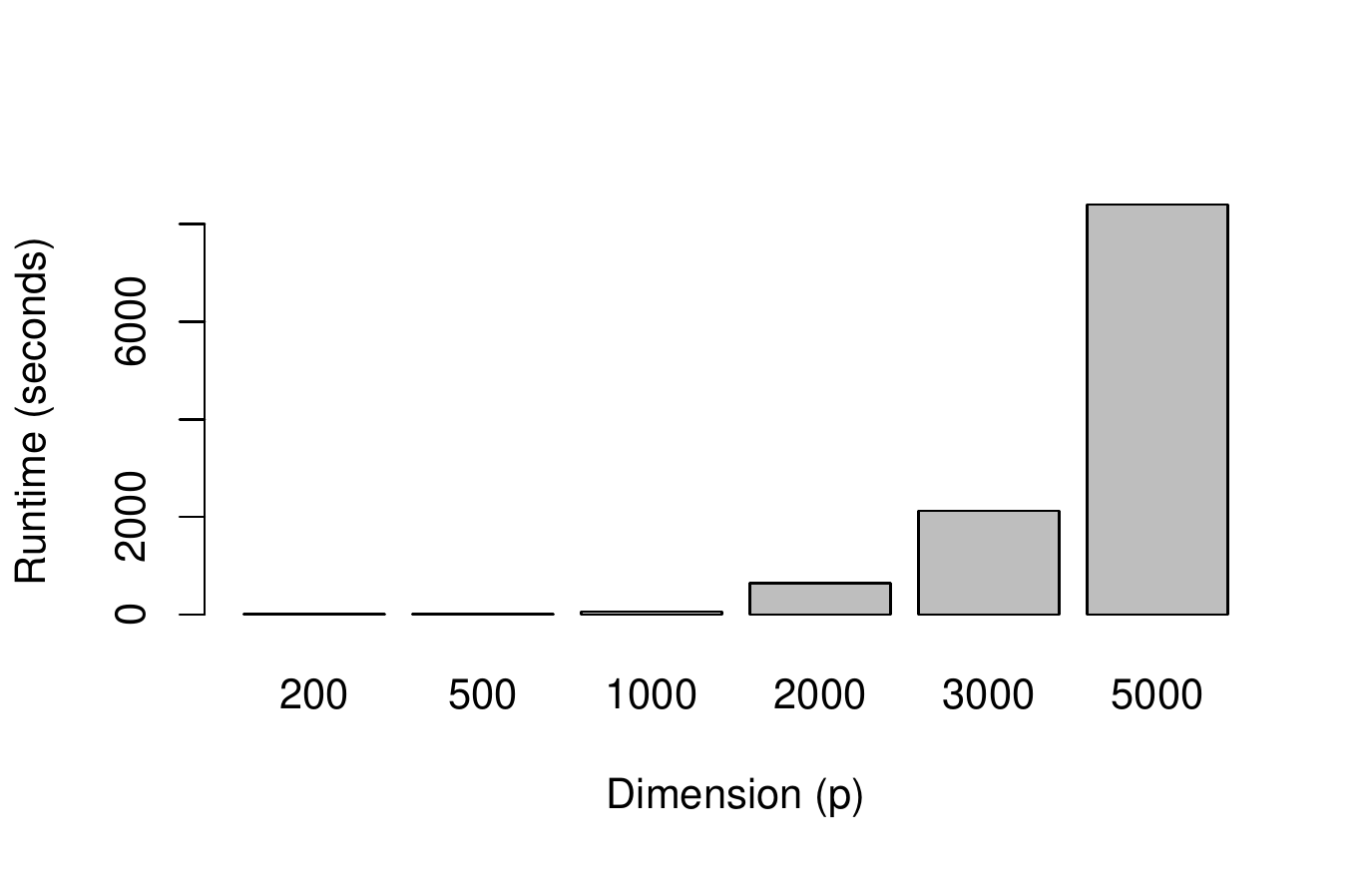}
\includegraphics[width = .51\textwidth]{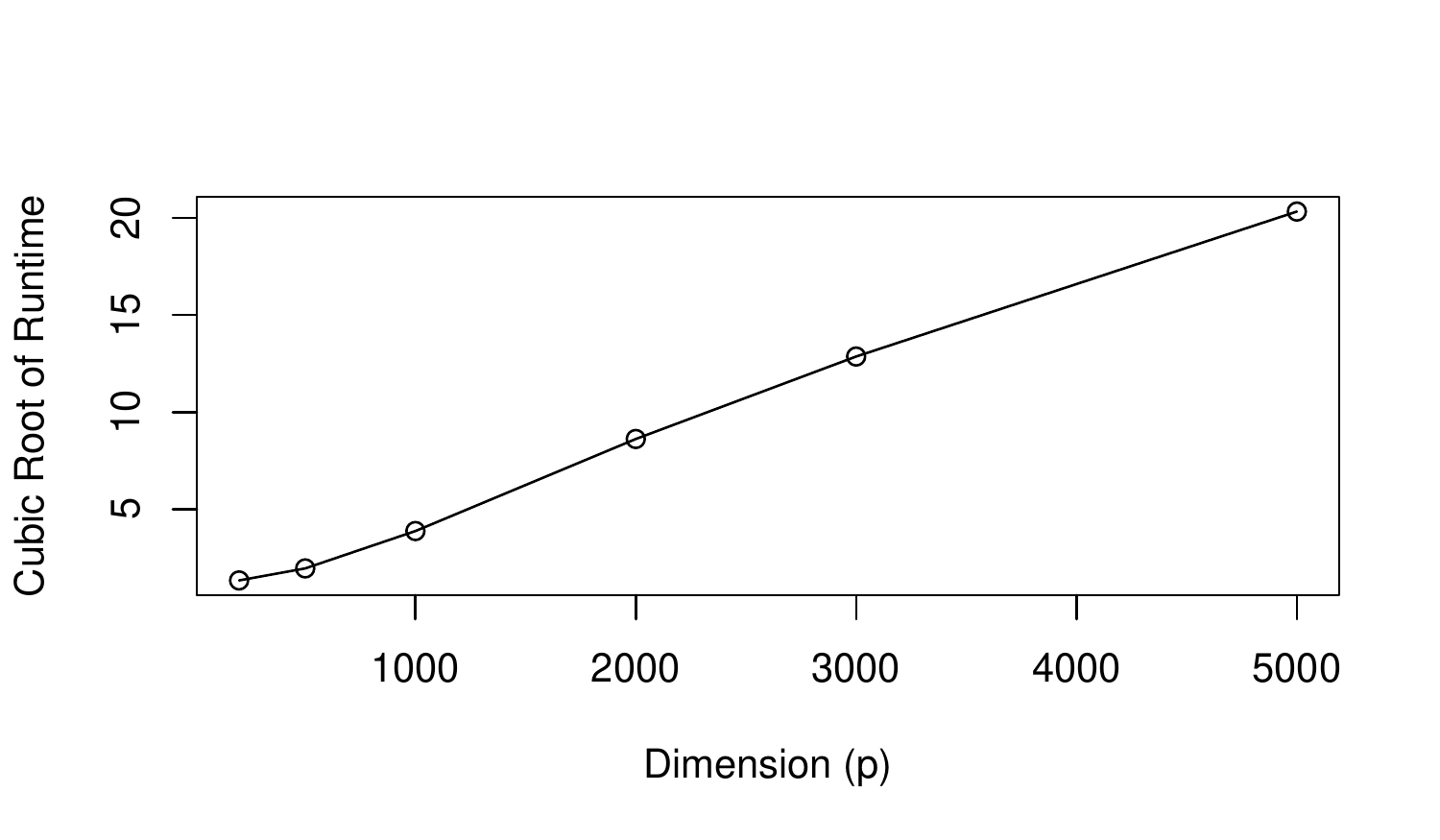} 
\caption{Runtime of  proximal distance algorithm as dimension grows.} 
\label{fig:runtime}
\end{figure}

\subsection{International Migration Data}\label{sec:migration}
Projecting international migration at the country-specific scale is important in shaping policy decisions that arise in social welfare and economic planning. Probabilistic projections are desirable in quantifying uncertainty in these ``barely predictable" global processes \citep{bijak2010}. Existing global models typically assume that forecast errors are uncorrelated across countries. Although modeling under the independence assumption may be well-calibrated for individual countries, ignoring correlations will yield under or over-estimates in projections. 

We consider international migration forecast data from the United Nations World Population Prospects (UN WPP) division. The dataset consist of net migration estimates every five years in each country from 1950 to 2010.  Following \cite{azose2016}, our goal is to estimate the correlation structure among forecast errors. The observations $\epsilon_t$ $(t=1, \ldots, 11)$ are residual vectors from an AR(1) model for net migration between all countries; the $\epsilon_t$ are assumed to be independent and identically distributed according to a multivariate normal distribution. We base inference on a small available sample \textcolor{black}{of $n=11$ measurements, seeking to estimate a correlation matrix $R$ with roughly $18,000$ entries generated by $p=191$ country pairs}. The Pearson sample correlation is known to degrade in such settings and suggest spurious correlations. \cite{azose2016} consider a Bayesian model that shrinks \textit{a priori} untrustworthy elements toward zero. This is achieved by penalizing country pairs that are far apart, do not share a colonial relationship, or occur in different regions. These penalties reflect the UN World Population Prospects partition of the globe into 22 regions based on geographical and cultural affinity. The authors employ a slight modification of the approach proposed by \cite{bien2011} to extract maximum \textit{a posteriori}  estimates.  
\cite{azose2016}  note that the method is slow on a problem of this size and renders cross-validation infeasible, instead choosing the $\ell_1$-penalty parameter $\lambda$ according to a manual heuristic. %
\begin{figure}[htbp]
\centering 
\hspace{-20pt}\includegraphics[width = 1.04\textwidth]{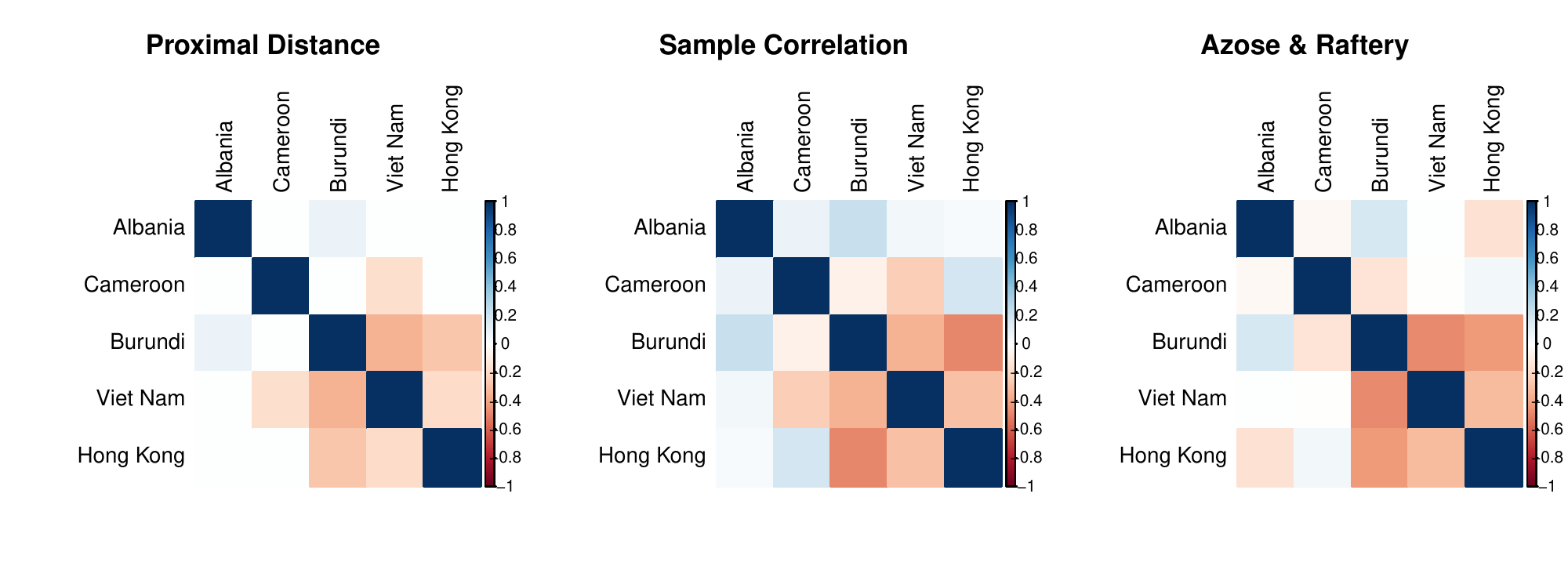}
\caption{Estimated correlations on a random subset of countries. The proximal distance algorithm results in a sensible sparsity pattern based on criteria suggested by \cite{azose2016}. Though not visibly obvious, their maximum \textit{a posteriori} estimate produces no sparse entries on this subset.} 
\label{fig:corrs}
\end{figure}

\textcolor{black}{In contrast, the analogous study using $5$-fold cross validation with the proximal distance method completes in under a minute on a standard laptop computer}. Our estimates deliver sparser solutions than the estimates of \cite{azose2016} under an $\ell_1$ penalty; we record $6145$ zeroes versus their $323$ zeroes. To illustrate the difference in estimates, we consider a random subset of five countries.  Despite the absence of prior knowledge, the proximal distance method reveals relationships that are qualitatively consistent with the criteria used by \cite{azose2016} in the design of their prior. As apparent in Figure \ref{fig:corrs}, the zero entries estimated under the proximal distance method correspond to country pairs that occur in different blocks of the UN partition. In contrast, the method of \cite{azose2016} produces small but nonzero entries for these pairs and, in general, does not produce a sparse solution. 
\begin{table}[htbp] \centering
\begin{tabular}{ l | c | c | c | c   }
 & $\hat\Ell $ & EBIC & BIC & AIC  \\ \hline
 Bayesian shrinkage & $\bf{-572.1}$  & $43208.1$ & 41591.1 & 34499.8 \\ 
Proximal distance & $-341.7$ & \bf 39701.3 & \bf 28091.3 & \bf 23216.6 
\end{tabular} 
\caption{ Comparison of the maximum a posteriori estimate by the Bayesian shrinkage approach of \cite{azose2016} and the proposed proximal distance estimate in terms of negative log-likelihood, extended BIC, BIC, and AIC. }
\label{tab:model}
\end{table}
To compare the quality of estimates quantitatively, we advocate the extended Bayesian information criterion (EBIC). Although this criterion tends to be more suitable in high-dimensional settings, we report in Table \ref{tab:model} the standard Akaike Information Criterion (AIC) and Bayesian Information Criterion (BIC) measures for completeness. As anticipated, the denser estimate of \cite{azose2016} achieves a lower negative log-likelihood on the data, but the measures accounting for model complexity favor our sparse solution. Given the limited amount of data, we hesitate to conclude that our estimate is definitively preferable. Indeed, the use of sensible prior knowledge in such a setting is prudent. Despite ignoring
\textit{a priori} information, it is noteworthy that our method is competitive with an ostensibly more tailored approach to the data at hand.

\begin{figure}[htbp]
\centering 
\includegraphics[width = 1.01\textwidth]{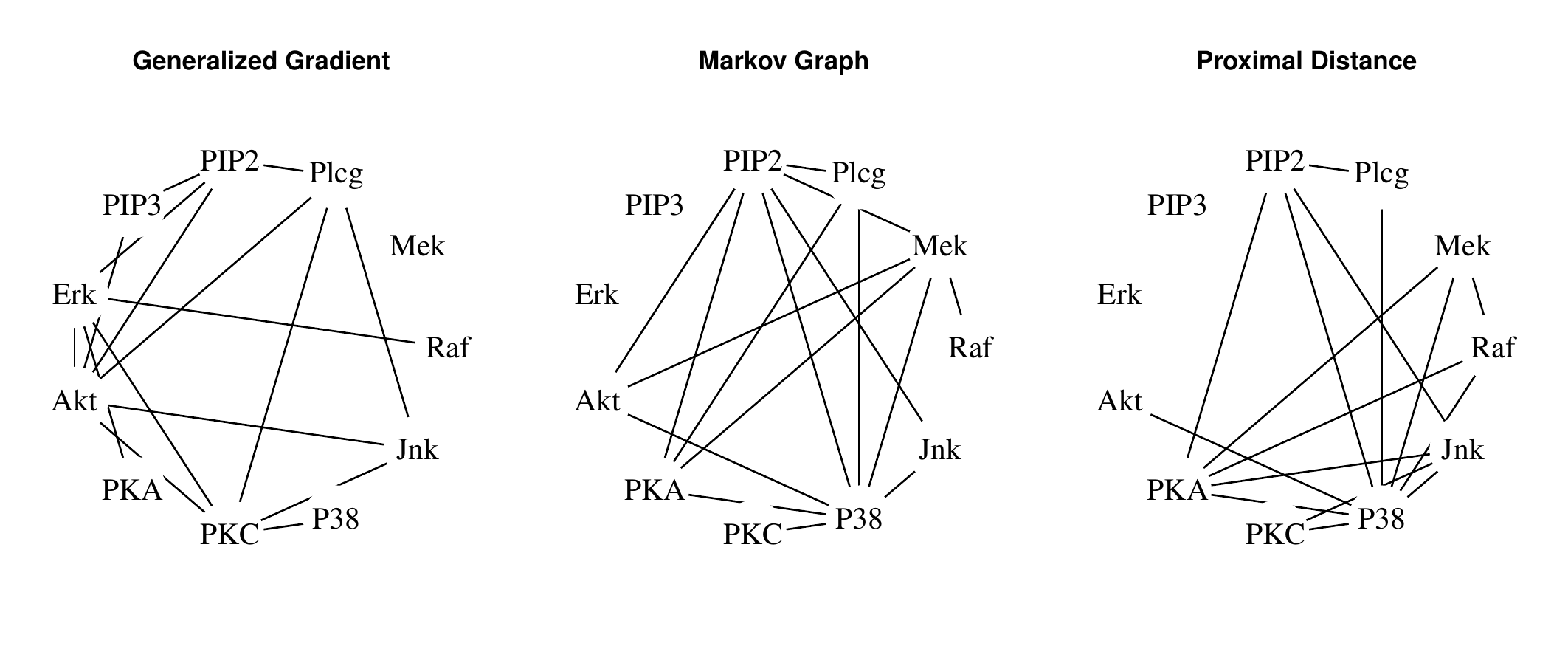} 
 \includegraphics[width = 1.01\textwidth]{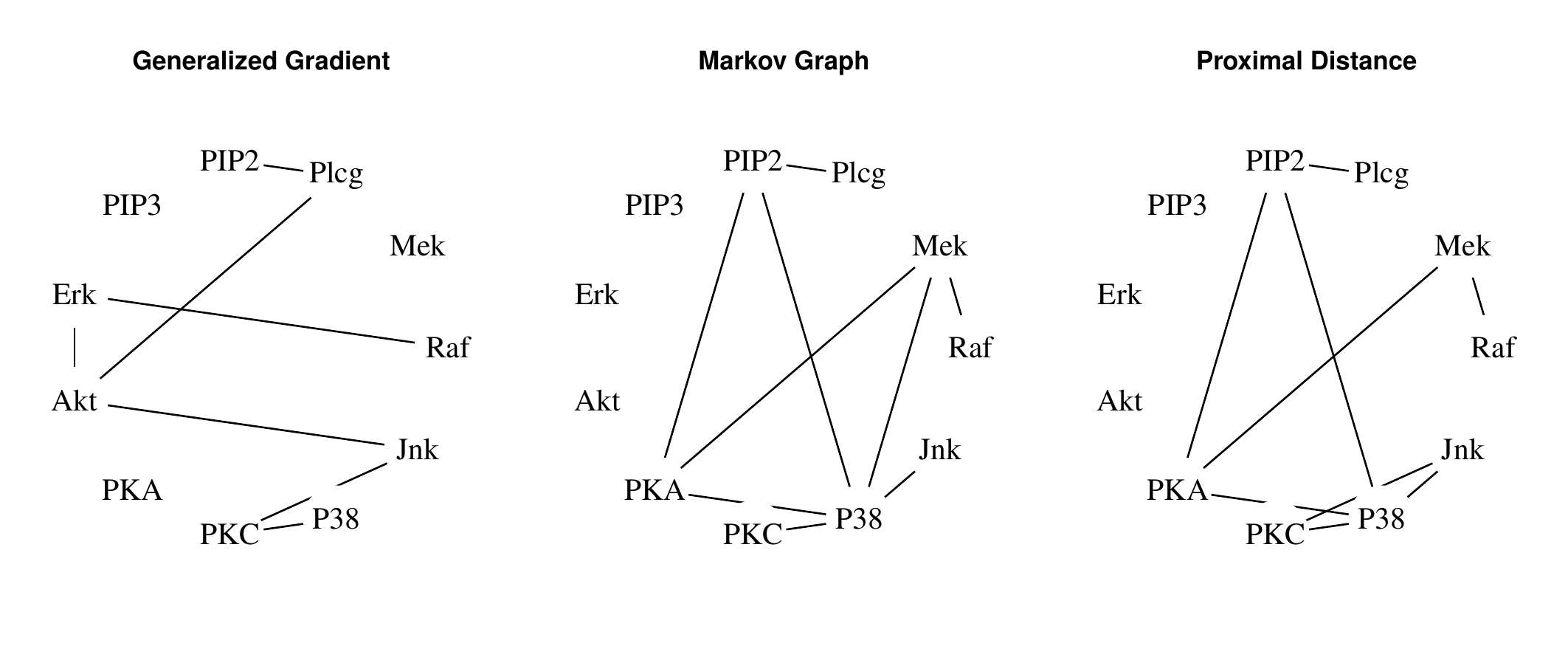}
\caption{Estimated covariance graphs under generalized gradient descent (left) and the proximal distance algorithm (right), compared to the generalized gradient estimate of the Markov graph (center). Top and bottom panels display two settings in which sparsity levels are matched between methods. In each case, the proximal distance algorithm produces more edges in common with the Markov graph.} 
\label{fig:cell}
\end{figure}

\subsection{Flow Cytometry data}\label{sec:cell}
Our final case study takes a closer look at the marginal and conditional dependency structures in a classic cell signaling study. We revisit the experiment studied by \cite{sachs2005} involving flow cytometry measurements on $p=11$ proteins and $n=7466$ cells. This dataset was previously analyzed in the original graphical lasso paper \citep{friedman2007} and in a study of $\ell_1$-penalized covariance estimation \citep{bien2011}. We produce two estimates of the conditional dependency or Markov graph using the graphical lasso with $k=9$ and $k=16$ edges. \cite{bien2011} apply their method to estimate the marginal dependency graph, which does not coincide with estimates of the Markov graph at matched sparsity levels. This is no surprise since the underlying models offer distinct interpretations.  A missing edge in the covariance graph tells us that the concentration of one protein gives no information about the concentration of the  other, whereas a missing edge in the Markov graph indicates that the concentration of one protein gives no information about the concentration of the other \textit{conditional} on all other concentrations. While this difference is crucial, our results suggest that the covariance graph may be more similar to the Markov graph than past studies based on $\ell_1$ penalties suggest.

Figure \ref{fig:cell} displays covariance graphs obtained by running the generalized gradient descent algorithm of \cite{bien2011} and our proximal distance algorithm at sparsity levels matched to the Markov graphs. It is visually clear that our estimate of $\widehat\Sigma$ shares more edges with the Markov graph. Although the true covariance graph and Markov graph do not necessarily coincide, these results suggest that the difference between the two in these data may be overstated due to $\ell_1$ shrinkage or convergence to a poor local minimum under generalized gradient descent. It is again difficult to produce a complete range of sparsity levels under an $\ell_1$ penalty. For instance, the generalized gradient estimate on the bottom row of Figure \ref{fig:cell} features one fewer edge than desired, though it yields the closest sparsity estimate before transitioning to $k=11$ edges over a grid search of mesh size $10^{-7}$ for the penalty constants. Even in the extreme case (not pictured) where penalty constants are chosen to yield only $k=1$ edge, the proximal distance algorithm and graphical lasso agree in producing the edge Mek---Raf, while the generalized gradient algorithm selects the sole edge Erk---Akt.  Once again we see that the proximal distance algorithm allows us to directly specify the sparsity level $k$, while $\ell_1$ penalization require tedious calibration to match the penalty constant $\lambda$ to $k$. This compact example  emphasizes both the computational advantages of the proximal distance algorithm and its ability to deliver dependable solutions uncontaminated by excess shrinkage.

\section{Discussion}
We propose a novel proximal distance method for estimating a sparse covariance matrix that does not appeal to global shrinkage or convex relaxation, and does not assume any known structure in the sparsity pattern. Our technique makes use of a natural and interpretable penalty based on the distance to the constraint set. We justify theoretically and showcase empirically the merits of the method. Building upon landmark work on this difficult non-convex problem by \cite{bien2011}, our contribution extends their core idea centered around majorization-minimization, but relies on a more favorable approximation that hugs the loss function more tightly at each iterate. We also employ a surprising solution method by drawing from ideas in the control theory literature. These changes substantially improve the stability, speed, and accuracy of sparse covariance matrix estimation.

In our proximal distance algorithm it is hard to avoid a computational complexity of $\mathcal{O}(p^3)$. For instance, formation of the left-hand side of equation \eqref{eq:sylvester} requires dense matrix inversion and multiplication.  One could possibly solve equation \eqref{eq:sylvester} by an iterative algorithm rather than the \cite{bartels1972} algorithm. 
For instance, the iteration scheme \[\bTheta_{j+1} = \frac{1}{\rho} \bD_k - \frac{1}{\rho} \bSigma_k^{-1} \bTheta_j \bSigma_k^{-1}\] converges
to $\bSigma_{k+1}$ provided that $\lVert \bSigma_k^{-1}\rVert_2^2 < \rho$.

We note that previous convergence results for proximal distance algorithms do not address non-convex sets. Though our theory handles the non-convex sparsity set $\mathcal{C}$, we fix the penalty constant $\rho$ in our analysis. This simplification is justified if we gradually increase $\rho$ and then fix its value, though there remain gaps that warrant further theoretical development of proximal distance algorithms. 
For example, how large should one take the resting value of $\rho$, and how quickly should one increment $\rho$ from its  initial value? That na\"ively using the same update schedule for $\rho$ works well empirically across the board should be considered an advantage. Nevertheless, a closer analysis of this behavior would be fruitful, and potentially crucial in other applications. 
 
Despite these gaps, the desirable theoretical properties and empirical prowess \textcolor{black}{of the proposed proximal distance algorithm suggest that the ideas} are applicable to a broad range of problems. The virtues of the principle include  versatility and elegance, avoidance of shrinkage, and ease of coding. Projection onto a closed set undergirds the principle. Fortunately, many projection operators are available in the literature, even for non-convex sets \citep{bauschke2011convex,beck2017}. \textcolor{black}{These successes encourage future work extending penalized likelihood methods for covariance estimation in the patternless sparsity setting. For instance in related problems, the local linear approximation algorithm succeeds in applying majorization-minimization for sparse estimation under alternate non-convex penalties such as SCAD \citep{zou2008one}. Exploring the extent to which our contributions can help tailor such approaches to sparse covariance estimation is noontrivial, and provides a fruitful avenue for future work.} We invite readers to help us advance proximal distance theory and devise their own applications of this valuable extension of the majorization-minimization principle.

\appendix
\section{Appendix \label{appendix1}}
\subsection*{Proof of Proposition \ref{step_halving_prop}}
It suffices to show that a small enough step size $s$ decreases $g_\rho(\bSigma \mid \bSigma_k)$. \textcolor{black}{Recall the form \eqref{eq:newton} which expresses $\widehat\bSigma = \bSigma_k -\bH_k^{-1} \nabla q_\rho(\bSigma_k \mid \bSigma_k)$, 
where $\bH_k$ is the scoring approximation obtained by taking the expected value of the second differential $d^2g_\rho(\bSigma \mid \bSigma_k)$}. Here we explicitly avoid writing $\bH_k$ as an unwieldy tensor, instead noting that it generates the positive definite quadratic form $\tr ( \bSigma_k^{-1} \bV \bSigma_k^{-1} \bV )$. In light of the identity $\nabla q_\rho(\bSigma_k \mid \bSigma_k)= \nabla g_\rho(\bSigma_k \mid \bSigma_k)$, the vector $\bv_k$ is a descent direction for $g_\rho(\bSigma \mid \bSigma_k)$ at $\bSigma_k$. Since  the cone of positive definite matrices is open, step-halving is also guaranteed to keep $\bSigma_{k+1}$ positive definite.

\subsection*{Proof of Proposition \ref{thm:convergence}}

To establish convergence, we invoke Zangwill's Global Convergence Theorem for descent algorithms \citep{luenberger1984,zangwill}. Recall that our algorithm map $A(\bSigma)$ may be set-valued because the projection $P_\mathcal{C}(\bSigma)$ onto the sparsity constraint set can be multi-valued. Denote the set of stationary points \eqref{stat_set}
of $A(\bSigma)$ by $\Gamma$. The theorem statement is reproduced in our notation below for convenience:

\begin{theorem}\label{thm:zangwill}
(Global Convergence Theorem) Consider the algorithm $A: X \rightarrow \mathcal{P}(X)$ defined by a point-to-set map and an initial point $\bSigma_0$. Let $\Gamma \subset X$ be a solution set and 
$\bSigma_{k+1} \in A(\bSigma_k)$ be a sequence generated by $A(\bSigma)$. Finally, assume that  
\begin{itemize}
\item[i)] all iterates $\bSigma_k$ are contained in a compact set $S \subset X$,
\item[ii)] there is a continuous function $h(\bSigma)$ such that
\begin{itemize}
\item[a)] if $\bSigma \notin \Gamma$, then $h(\bTheta) < h(\bSigma)$ for all $\bTheta \in A(\bSigma)$,
\item[b)] if $\bSigma \in \Gamma$, then $h(\bTheta) \leq h(\bSigma)$ for all $\bTheta \in A(\bSigma)$,
\end{itemize}
\item[iii)] the mapping $A$ is closed at points outside of $\Gamma$.
\end{itemize}
Then the sequence $\bSigma_k$ possesses convergent subsequences, and the corresponding limits belong to the solution set.
\end{theorem}

We begin by proving the coercivity of $h_\rho(\bSigma)$, which will imply that the sequence $\bSigma_k$ is contained in a compact set. Note  even if $\bS$ is singular, running our method instead on $\tilde\bS = \bS + \delta \bI$ for arbitrarily small $\delta$ suffices for the theory to hold. Doing so is reasonable as it is a strictly weaker assumption than relaxing the entire constraint $\bSigma \succ {\bf 0}$ to the set $\bSigma \succeq \delta \bI$ \citep{bien2011}.

\begin{lemma}\label{lem:coercive}
The objective function $h_\rho(\bSigma)$ of our model is coercive whenever the sample covariance matrix $\bS$ is nonsingular.
\end{lemma}
\begin{proof} Since the penalty is nonnegative, it suffices to prove that $f(\bSigma)=\ln \det \bSigma + \text{tr}(\bSigma^{-1}S)$ is coercive. Let the singular values of $\bSigma$ be denoted $\sigma_1 \geq \sigma_2 \geq \ldots > 0$, and let the singular values of $\bS$ be denoted $s_1 \geq s_2 \geq \ldots >0$. It is clear that $\lVert \bSigma \rVert \rightarrow \infty$ if and only if at least one $\sigma_i\rightarrow \infty$ and that $\lVert \bSigma^{-1} \rVert \rightarrow \infty$ if and only if at least one $\sigma_i \rightarrow 0$. The matrix analogue of the Cauchy-Schwarz inequality due to Von Neumann and Fan tells us that $\text{tr}(\bSigma^{-1}\bS) \geq \sum_i s_i/\sigma_i$. We also have $\ln \det \bSigma = \sum_i \ln \sigma_i$. Now consider the sum $r(\bsigma)=\sum_i (\ln \sigma_i + s_i/\sigma_i)$, which bounds $f(\bSigma)$ below. Since each summand satisfies $$ \min_{\sigma_i} \Big(\ln \sigma_i + \frac{s_i}{\sigma_i}\Big) \geq \ln s_i + 1,$$ $r(\bsigma)$ obviously tends to $\infty$ if and only if any $\sigma_i$ tends to $0$ or $\infty$. Equivalently,
$f(\bSigma)$ tends to $\infty$ if and only if either $\lVert \bSigma \rVert$ or $\lVert \bSigma^{-1} \rVert$ tends to $\infty$.
\end{proof}

This proof shows that if we set $h_\rho(\bSigma) = \infty$ where $\bSigma$ fails to be positive definite, then $h_\rho(\bSigma)$ is continuous. We will adopt this convention in defining the update
$\bSigma_{k+1} = \bSigma_k + \eta_k \bv_k \in A(\bSigma_k)$ via the
choice
$$\eta_k = \argmin_{\eta \in [0,1]} g_\rho(\bSigma_k + \eta \bv_k \mid \bSigma_k).$$ 
Before proving the next lemma, recall that the surrogate $q_\rho(\bSigma \mid \bSigma_k)$ is minimized by $\widehat\bSigma~=~\bSigma_k~+~\bv_k$, where $\bv_k = -\bH_k^{-1} \nabla q_\rho(\bSigma_k \mid \bSigma_k)$ and $\bH_k$ is the approximate second differential generating the quadratic form 
$\bV \mapsto \tr ( \bSigma_k^{-1} \bV \bSigma_k^{-1} \bV )$. 
Elements of the solution set $\Gamma$ of Zangwill's theorem are characterized by the stationarity condition \eqref{stat_set} for some $\bTheta \in P_\mathcal{C}(\bSigma)$. 

\begin{lemma}\label{lem:zangwill}
Some point $\bTheta \in A(\bSigma)$ decreases our objective $h_\rho(\bSigma)$ and strictly so
when $\bSigma \notin \Gamma$. Furthermore, the algorithm map $A(\bSigma)$ remains within a compact set and is closed outside $\Gamma$.
\end{lemma}
\begin{proof}
By definition the algorithm map decreases $h_\rho(\bSigma)$. If $\bSigma$ falls outside $\Gamma$, then any associated search direction $\bv$ be expressed as $\bv = -\bH^{-1} \bu$, where $\bH$ is positive definite and $\bu = \nabla q_\rho(\bSigma \mid \bSigma)$ is nontrivial for any choice of $\bTheta \in P_\mathcal{C}(\bSigma)$. Because
\begin{eqnarray*}d_{\bv} g_\rho(\bSigma \mid \bSigma) = d_{\bv}q_\rho(\bSigma \mid \bSigma)
= - \bu^T \bH^{-1}\bu < 0 ,
\end{eqnarray*} 
it follows that $g_\rho(\bSigma \mid \bSigma)$ can be strictly decreased by moving in the direction $\bv$. Hence, the objective $h_\rho(\bSigma)$ can be strictly decreased. To prove compactness, note that $h_\rho(\bSigma)$ is both continuous and coercive. Hence, its sub-level sets $\{ \bSigma : h_\rho(\bSigma) \le c\}$ are compact. Given that the algorithm decreases $h_\rho(\bSigma)$, all iterates remain within the compact set $\{\bSigma : h_\rho(\bSigma) \le h_\rho(\bSigma_0)\}$. 

To prove closedness, consider a sequence $\bSigma_k$ with limit $\bSigma$ and a corresponding
sequence $\bTheta_k \in A(\bSigma_k)$ with limit 
$\bTheta \not\in \Gamma$. If $f(\bSigma)$ is the loss function, then
$\bu_k=\nabla f(\bSigma_k)+\rho(\bSigma_k-\bTheta_k)$, where 
$\bTheta_k \in  P_{\mathcal{C}}(\bSigma_k)$.
The lack of continuity of the projection operator hinders taking limits. However, since there are only a finite number of sparsity index sets, one of these sets must be chosen infinitely often along the sequence 
$\bTheta_k$. Replace the sequences $\bSigma_k$ and $\bTheta_k$ by the subsequence where this occurs. One can now invoke the continuity of the projection operator and conclude that $\bTheta = \lim_{k \to \infty}\bTheta_k$ exists. It follows that 
$$\bv=\lim_{k \to \infty} \bv_k  = -\bH^{-1}[\nabla f(\bSigma)+\rho(\bSigma-\bTheta)]$$
also exists with $\bTheta \in  P_{\mathcal{C}}(\bSigma)$. Furthermore, 
$\bv \neq  {\bf 0}$ since $\bSigma \not\in \Gamma$. The step-length
sequence $\eta_k$ also has a limit $\eta$ defined by
$$\eta = \lim_{k \to \infty}\frac{\lVert \bTheta_k-\bSigma_k \rVert_2}
{\lVert \bH_k^{-1}\bu_k \rVert_2}
= \frac{\lVert\bTheta-\bSigma\rVert_2}{\lVert\bv\rVert_2}.$$
It remains to prove that $\bTheta = \bSigma + \eta \bv$ is optimal.
Fortunately, this follows by taking limits in the inequality
$g(\bSigma_k+\eta_k \bv_k) \le g(\bSigma_k+\mu \bv_k)$ valid for all $\mu \in [0,1]$.
%
\end{proof}

Now we are ready to prove Theorem \ref{thm:convergence} by a direct application of Zangwill's theorem.
\begin{proof}
The sub-level set $S_{h_\rho}(\bSigma_0)=\{\bSigma: h_\rho(\bSigma) \le h_\rho(\bSigma_0)\}$ is compact, and by Lemmas \ref{lem:coercive} and \ref{lem:zangwill}, all iterates $\bSigma_{k+1} \in A(\bSigma_k)$ lie in $S_{h_\rho}(\bSigma_0)$. These lemmas further show that 
a) $\bSigma_k \succ {\bf 0}$ for every $k$, b) $h_\rho(\bSigma_k)$ is continuous,
c) $h_\rho(\bTheta) \le h_\rho(\bSigma)$ for all $\bTheta \in A(\bSigma)$, and d)
equality is strict here when $\bSigma \notin \Gamma$. Furthermore,
the algorithm map $A(\bSigma)$ is closed outside $\Gamma$, the
set of stationary points. 
Therefore, Theorem \ref{thm:zangwill} applies, and every convergent subsequence of $\bSigma_k$ is a stationary point.
\end{proof}

\subsection*{Additional Simulation Details}
\textcolor{black}{
The experimental design in the first set of simulations are a direct reproduction of those in \citep{bien2011}. The analogous results presented in term of receiver operating characteristic curves appear below in Figure \ref{fig:roc}. Any simulated datasets that fail to produce a positive definite ground truth covariance matrix are re-simulated. Next, all methods are seeded and run on the same synthetic datasets with matched relative tolerance. In all results, the penalty parameter $\lambda$ for competing methods and the sparsity level $k$ for our proposed method are selected via $5$-fold cross validation with respect to Frobenius loss over a vector of $40$ possible values, calibrated so that best values do not occur on either boundary of the vector. This follows the recommendation in the implementations of those methods in the R packages \texttt{CVTuningCov} and \texttt{PDSCE}. We remark that cross-validation with respect to entropy loss was more favorable to our proposed method, though reported results in Tables 1---3 are cross-validated under Frobenius loss to offer a conservative comparison against peer methods.} The initial value of the parameter $\rho$ is set to $0.1$ in all cases considered and is not considered a tuning parameter.

\begin{figure}[htbp]
\centering\includegraphics[width = .32\textwidth]{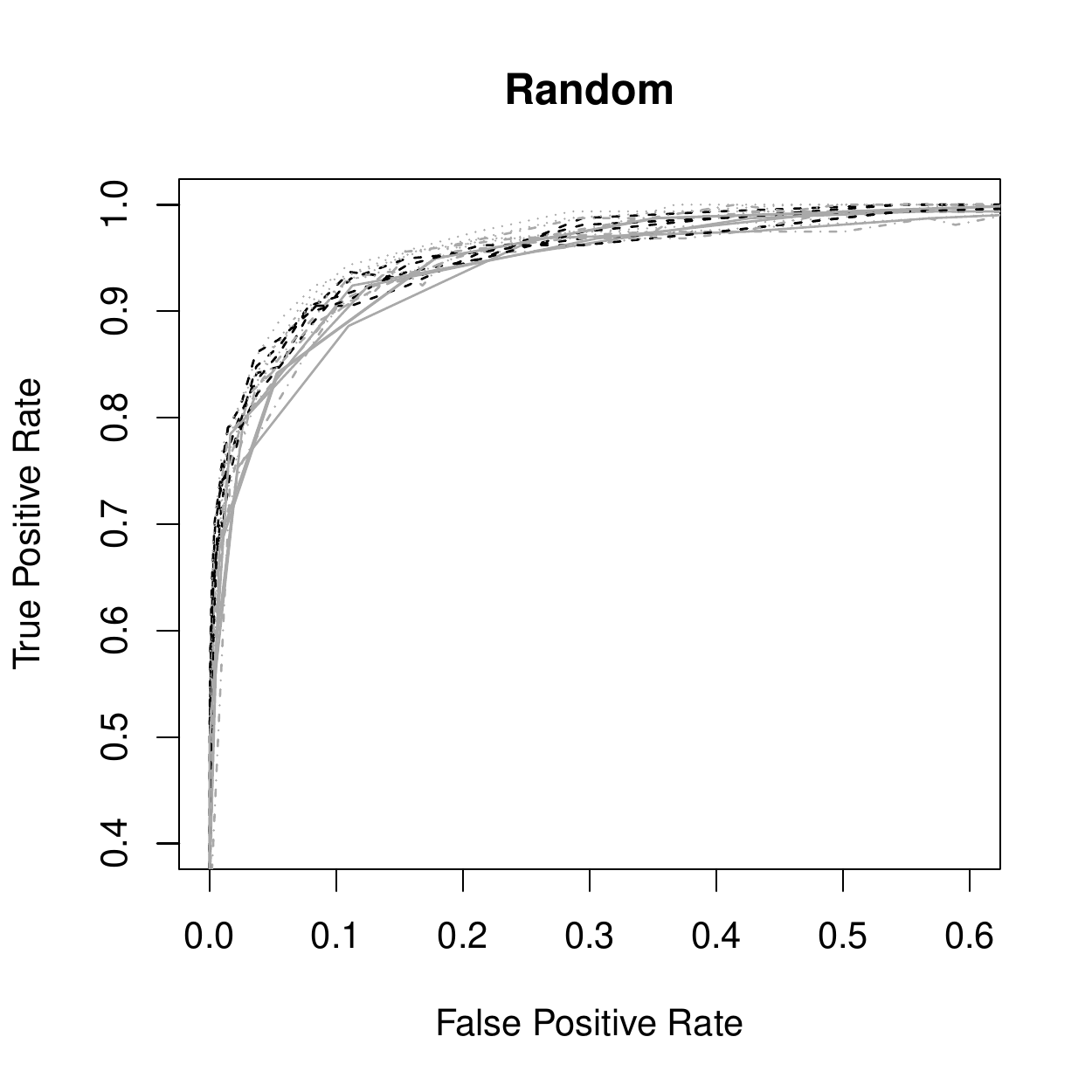}
\includegraphics[width = .32\textwidth]{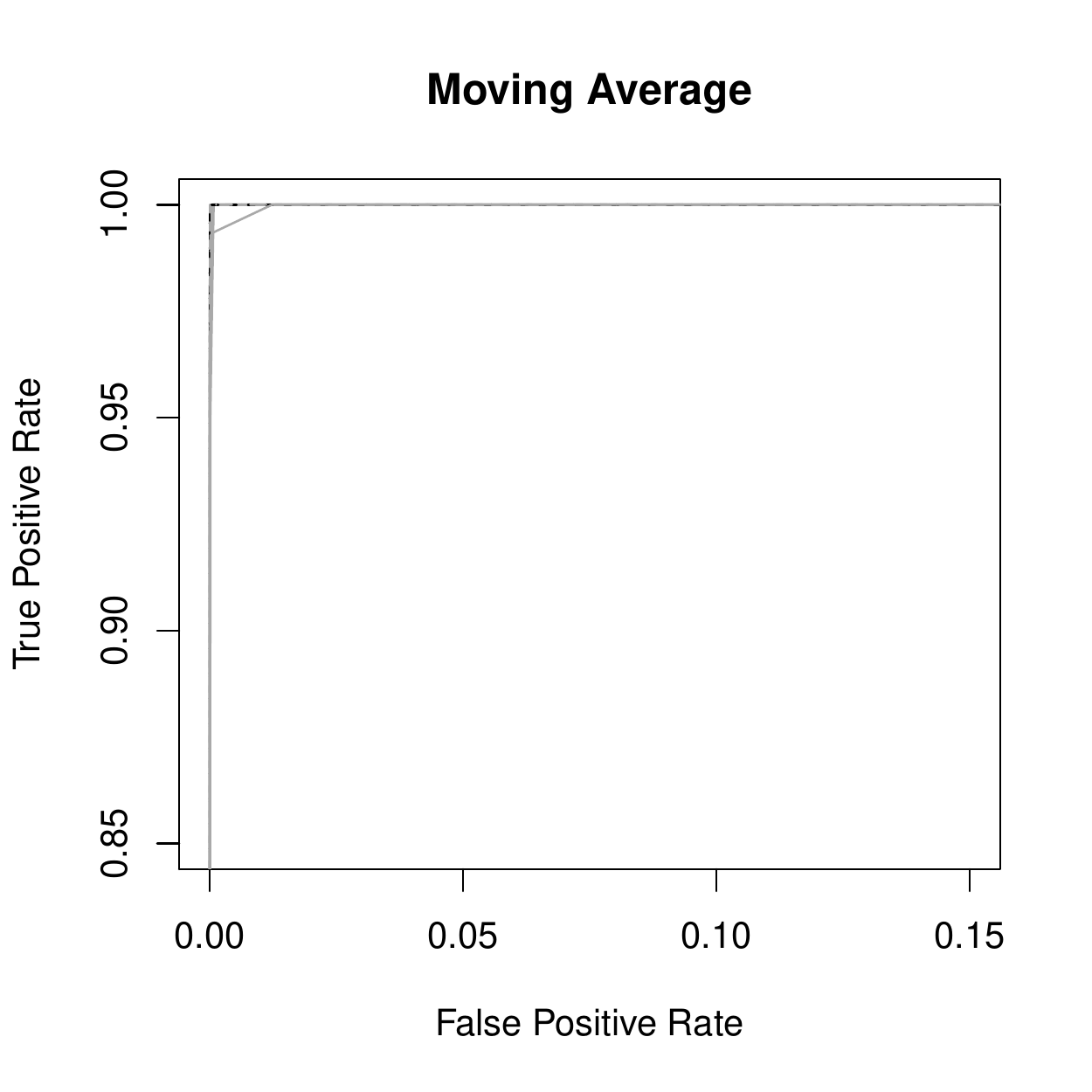} 
\includegraphics[width = .32\textwidth]{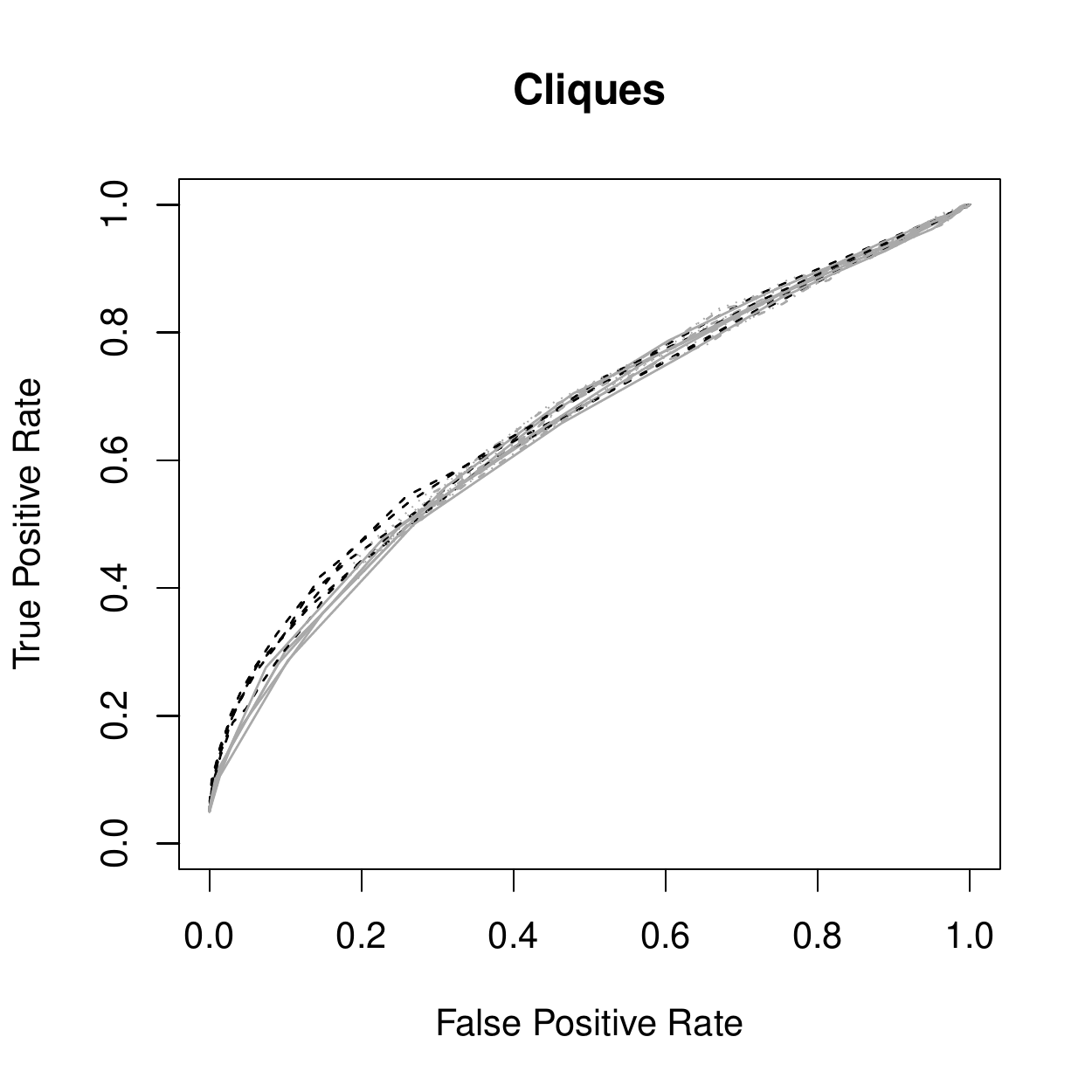}
\caption{\textcolor{black}{Receiver operating characteristic curves corresponding to the simulation study and display conventions of Figure 1.}}\label{fig:roc}
\end{figure}

\bibliographystyle{plainnat}
\bibliography{covariance}

\end{document}